\documentclass{article}

\usepackage{lipsum}
\usepackage{amsfonts}
\usepackage{graphicx}
\usepackage{epstopdf}
\usepackage{algorithmic}
\usepackage{multibib}
\usepackage[affil-it]{authblk}
 \usepackage[sort,nocompress]{cite}
\ifpdf
  \DeclareGraphicsExtensions{.eps,.pdf,.png,.jpg}
\else
  \DeclareGraphicsExtensions{.eps}
\fi

\usepackage[T1]{fontenc}
\usepackage[utf8]{inputenc}
\usepackage{amsmath}
\usepackage{amssymb}
\usepackage{mathtools}
\usepackage{multirow}
\usepackage{changepage}
\usepackage{pdflscape}
\usepackage{afterpage}
\usepackage{braket}
\usepackage{url}
\usepackage{mathdots}
\usepackage[export]{adjustbox}
\usepackage{booktabs}
\usepackage{siunitx}
\usepackage{caption}
\usepackage{subcaption}
\usepackage{cleveref}
\usepackage{color} 
\usepackage[vlined, ruled, boxed]{algorithm2e}
\usepackage{bm}
\usepackage{diagbox}
\usepackage{soul}
\usepackage{mathdots}
\usepackage{textcomp}

\usepackage{makecell}
\usepackage{tabularx,booktabs}
\usepackage{tikz}
\usepackage[customcolors]{hf-tikz}

\newcommand{\llbracket}{{[\![}}
\newcommand{\rrbracket}{{]\!]}}

\newcommand{\rc}[1]{\tikzmarkin[set fill color=red!50, set border color=red!50]{#1}(0.3,-0.1)(-0.15,0.3)\tikzmarkend{#1}}
\newcommand{\gc}[1]{\tikzmarkin[set fill color=green!50!lime!30, set border color=green!50!lime!30]{#1}(0.3,-0.1)(-0.15,0.3)}
\newcommand{\pc}[1]{\tikzmarkin[set fill color=purple!30, set border color=purple!30]{#1}(0.3,-0.1)(-0.15,0.3)\tikzmarkend{#1}}
\newcommand{\egc}[1]{\tikzmarkend{#1}}

\DeclarePairedDelimiter{\ceil}{\lceil}{\rceil}
\DeclarePairedDelimiter{\floor}{\lfloor}{\rfloor}

\makeatletter

\DeclareUnicodeCharacter{00A0}{ }

\SetKwProg{Fn}{Function}{}{}

\makeatother

\usepackage[english]{babel}
\usepackage[normalem]{ulem}
\usepackage{amsthm}

\newcolumntype{Y}{>{\raggedright}X}

\newtheorem{mythm}{Theorem}
\newtheorem{myprop}{Property}

\newtheorem{mypropo}{Proposition}

\author[1]{Timothée Goubault de Brugière }
\author[2]{Simon Martiel}
\affil[1]{Université de Lorraine, CNRS, Inria, LORIA, F-54000 Nancy, France}
\affil[2]{Atos Quantum Lab, Les Clayes-sous-bois, France}

\date{\today}

\title{Shallower CNOT circuits on realistic quantum hardware}

\begin{document}

\maketitle

\begin{abstract}

We focus on the depth optimization of CNOT circuits on hardwares with limited connectivity. We adapt the algorithm from Kutin \textit{et al.} \cite{DBLP:journals/cjtcs/KutinMS07} that implements any $n$-qubit CNOT circuit in depth at most $5n$ on a Linear Nearest Neighbour (LNN) architecture. Our proposal is a block version of Kutin \textit{et al.}'s algorithm that is scalable with the number of interactions available in the hardware: the more interactions we have the less the depth. We derive better theoretical upper bounds and we provide a simple implementation of the algorithm. Overall, we achieve better depth complexity for CNOT circuits on some realistic quantum hardware like a grid or a ladder. For instance the execution of a $n$-qubit CNOT circuit on a grid can be done in depth $4n$.

\end{abstract}

\section{Introduction}

Quantum decoherence is a major obstacle to the scaling of quantum computing. It is very hard to maintain the qubits isolated from the environment during a calculation and once the qubits interact with external elements the result of the current calculation is lost. The \textit{decoherence time} is used to designate this limited amount of available computing time. This time is given by the hardware and adds an additional constraint at the software level and more precisely during compilation: the sequence of instructions generated by the compiler for the machine must be able to be executed in a sufficiently short time. In the \textit{quantum circuit} model, this is equivalent to say that the \textit{depth} of the circuit must be as low as possible.

In this article we tackle the depth optimization of a specific class of quantum operators, namely the linear reversible operators or equivalently the CNOT circuits. CNOT circuits are themselves a subclass of the so-called Clifford circuits which play a major role in many different area of quantum computation such as quantum error correction \cite{gottesman1997stabilizer}, randomized benchmarking protocols \cite{knill2008randomized,magesan2011scalable}, quantum state distillation \cite{bravyi2005universal, knill2005quantum}. CNOT circuits are also used in other classes of quantum circuits, for instance in the synthesis of phase polynomials \cite{amy2018controlled}. More directly, the optimization of CNOT circuits has also been useful for the optimization of more general quantum circuits like arithmetic circuits \cite{de2021gaussian,de2021reducing}. 

We take into account some architectural constraint between the qubits. The CNOT gate is a two-qubit gate, if two qubits are not close enough in the hardware their interactions is not physically achievable and a CNOT gate between these two qubits cannot be executed. This limits the pair of qubits on which one can apply a CNOT gate and adds even more constraints to the compiler. In an ideal case, all the qubits are connected: then we talk about \textit{full connectivity} or \textit{all-to-all connectivity}, otherwise we talk about \textit{partial} or \textit{constrained} connectivity between the qubits.

The optimization of CNOT circuits have attracted a lot of attention in recent years. Two metrics are generally used to evaluate the cost of running a CNOT circuit: its size or its depth. There are also two types of connectivity: full connectivity and partial connectivity. In total, this gives four cases to be treated. If we can find recent works that optimize the size in full connectivity \cite{de2021gaussian,de2021decoding} and in partial connectivity \cite{de2021decoding,tang2020efficient,DBLP:journals/qic/KissingerG20,nash2020quantum}, and works that optimize the depth in full connectivity \cite{de2021reducing,maslov2022depth,jiang2020optimal}, nothing to our knowledge has been recently proposed to optimize the depth of CNOT circuits with architectural constraints. As far as we know, only two 15-year old works proposed implementations of CNOT circuits in the specific case of a Linear Nearest Neighbour (LNN) architecture \cite{DBLP:journals/cjtcs/KutinMS07,maslov2007linear}. The main result is that $n$-qubit CNOT circuits can be implemented in depth at most $5n$ for the LNN architecture. No improvement of this result nor extensions to other architectures were proposed since.

We propose a block generalization of the algorithm proposed in \cite{DBLP:journals/cjtcs/KutinMS07}. Our algorithm works for any architecture where the qubits can be packed into groups of equal size such that the groups are arranged as a line. So for instance with groups of size 2 our algorithm can treat the case of a ladder like the architecture IBM QX5. With groups of size 4 we can deal with the grid. While not being universal for any qubit connectivity, our algorithm is versatile enough to work for realistic quantum architectures. 

We show that the algorithm skeleton does not depend on a specific architecture or the group size. Our algorithm consists in a series of small problems involving boolean matrices to zero: these problems are the atoms of our algorithm and this is where the block size and the actual architecture specify the constraints with which we have to solve them. We propose several strategies to solve these problems for different problem sizes and different architectures. Overall, we show that any $n$-qubit CNOT circuit can be executed in depth at most $4n$ in the case of the ladder. For the grid, a depth of $4n$ is sufficient and a depth of $15n/4$ is enough if we add the diagonals. 

The structure of the article is a follows: in Section~\ref{sec::background} we give a brief background about CNOT circuits synthesis and we give a tuned formulation of Kutin \textit{et al.}'s algorithm proposed in \cite{DBLP:journals/cjtcs/KutinMS07} for the LNN architecture. Then in Section~\ref{sec::block} we propose our block extension of Kutin \textit{et al.}'s algorithm. We detail the general structure of the algorithm, independent of the architecture and the block size, and the subproblems we have to solve. Then in Section~\ref{sec::practical} we give several ways to solve these subproblems for different cases of block sizes and architectures. We conclude in Section~\ref{sec::conclusion}.

\section{Background and Kutin \textit{et al.}'s algorithm} \label{sec::background}

\subsection{Background}

The CNOT gate is a classical reversible operator. It applies a NOT gate on a \textit{target} qubit if the value of a \textit{control} qubit is True. This is equivalent to writing 
\[ CNOT(x_1, x_2) = (x_1, x_1 \oplus x_2) \]
where $x_1$, resp. $x_2$, is the logical input value of the control, resp. target, qubit and $\oplus$ is the XOR operator. 

By extension any CNOT circuit on $n$ qubits applied to a bitstring $x = (x_1, \hdots, x_n)$ outputs a bitstring $y = (y_1, \hdots, y_n)$ where each $y_i$ is a linear combination of the $x_i$'s. In other words, 
\[ y = Ax = \begin{bmatrix} A_{11} x_1 \oplus A_{12} x_2 \oplus \hdots \oplus A_{1n} x_n \\ \vdots \\ A_{n1} x_1 \oplus A_{n2} x_2 \oplus \hdots \oplus A_{nn} x_n \end{bmatrix}. \]

$A \in \mathbb{F}_2^{n \times n}$ completely characterizes the functionality of the CNOT circuit. By reversibility of the operator, $A$ is necessarily invertible. Given a CNOT circuit implementing an operator $A$, outputting a bitstring $y=Ax$, executing an additional CNOT gate with control $i$ and target $j$ will perform the operation 
\[ CNOT(y_i, y_j) = (y_i, y_i \oplus y_j) = A'x \]
where $A'$ is given from $A$ by adding the row $i$ to the row $j$. We write $A' = E_{ij}A$ where $E_{ij} = I \oplus e_{ji}$ and $e_{ji}$ is zero everywhere expect in the entry $(j,i)$, and we note that $E_{ij}^{-1} = E_{ij}$.

The simulation of CNOT circuits is therefore efficient and its optimization affordable for a compiler. We just showed an equivalence between applying a CNOT gate and performing an elementary row operation on the matrix operator $A$. If one finds a suitable sequence of $N$ row operations $(E_{i_k, j_k})_k$ such that 
\[ \left( E_{i_N,j_N} \times E_{i_{N-1}, j_{N-1}} \times \hdots \times E_{i_1, j_1} \right) A = I \] 
then 
\[ A = \prod_{k=1}^N E_{i_k, j_k} \]
and one gets a direct implementation of $A$ as a CNOT circuit by reading the different elementary row operations. Such decomposition is always possible, for instance any Gaussian elimination algorithm works.

To summarize the synthesis of a CNOT circuit is equivalent to reducing an invertible boolean matrix to the identity with the use of elementary row operations. One is interested in an \textit{efficient} algorithm to do such reduction. To evaluate the efficiency of a synthesis algorithm, two metrics are used: 
\begin{itemize}
    \item the number of CNOT gates in the circuit, in other words the number of row operations used,
    \item the depth of the circuit. This corresponds to the number of time steps needed to execute the circuit given that two gates that act on distinct qubits can be executed simultaneously.
\end{itemize}

Moreover, the CNOT is a two-qubit gate, and therefore requires the interactions between two qubits to be executed. This might not be always physically possible to perform such interactions due to hardware constraints. We talk about graph connectivity to encode the available interactions: the nodes are the qubits and an edge indicates that the two qubits can interact and that a CNOT gate can be executed. When the graph is complete, all qubits are connected and we have a full qubit connectivity. Otherwise, the connectivity is said to be partial. We give some examples of realistic and existing qubit connectivities in Figure~\ref{fig::connectivities}.

The first algorithm improving the Gaussian elimination algorithm was the Patel-Markov-Hayes (PMH) algorithm \cite{patel2008optimal}. It works in the case of a full qubit connectivity and generates circuits of size $O(n^2/\log_2(n))$ where $n$ is the number of qubits. Then several algorithms were proposed in recent years to improve the PMH algorithm \cite{de2021gaussian,de2021decoding} and extensions to partial connectivities were also proposed \cite{de2021decoding,tang2020efficient,DBLP:journals/qic/KissingerG20,nash2020quantum}. 

For the depth optimization, surprisingly, the first works treated directly the case of a LNN architecture \cite{DBLP:journals/cjtcs/KutinMS07,maslov2007linear}. The main result is that $n$-qubit CNOT circuit can be executed in depth at most $5n$ \cite{DBLP:journals/cjtcs/KutinMS07}. More surprisingly, to our knowledge, no other work was proposed to either improve the complexity in the LNN case or to extend it to other connectivities. For a full qubit connectivity, several algorithms were proposed in recent years \cite{de2021reducing,maslov2022depth,jiang2020optimal}, notably the asymptotic optimum of $O(n/\log_2(n))$ is achievable \cite{jiang2020optimal}. 

For completeness and for clarity, we now detail Kutin \textit{et al.}'s algorithm achieving a depth of $5n$ on a line. Our block extension is natural with a suitable formulation of this algorithm.

\begin{figure}
\begin{subfigure}[b]{0.45\textwidth}
    \centering
    \includegraphics[scale=0.35]{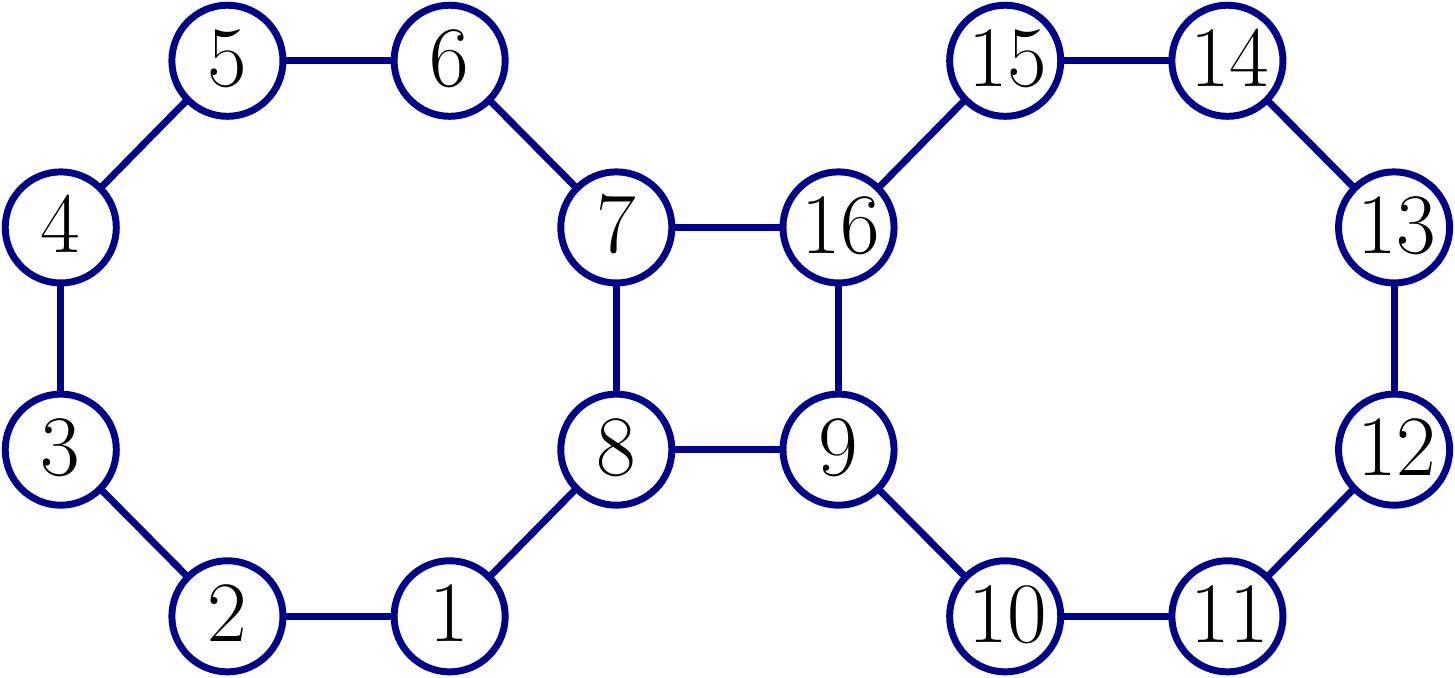}
    \caption{Rigetti 16Q-Aspen}
\end{subfigure}
~
\begin{subfigure}[b]{0.45\textwidth}
    \centering
    \includegraphics[scale=0.35]{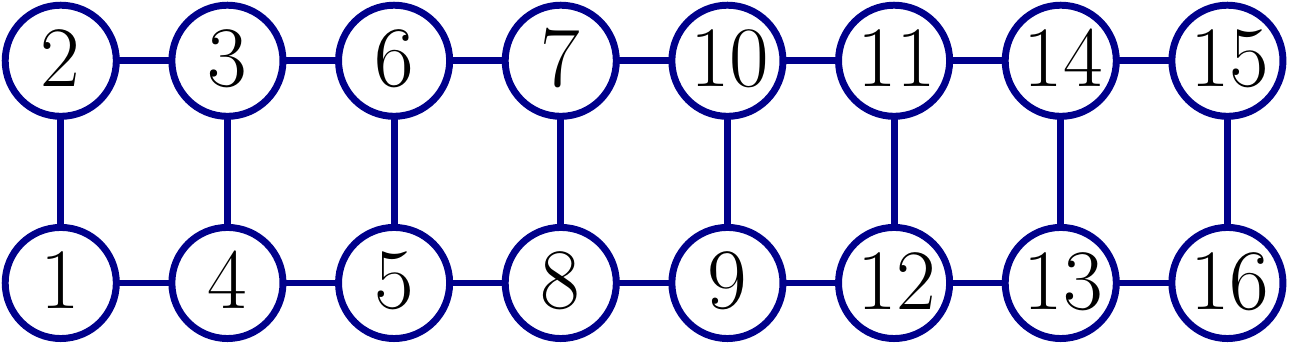}
    \caption{IBM QX5}
\end{subfigure}

\begin{subfigure}[b]{0.45\textwidth}
    \centering
    \includegraphics[scale=0.35]{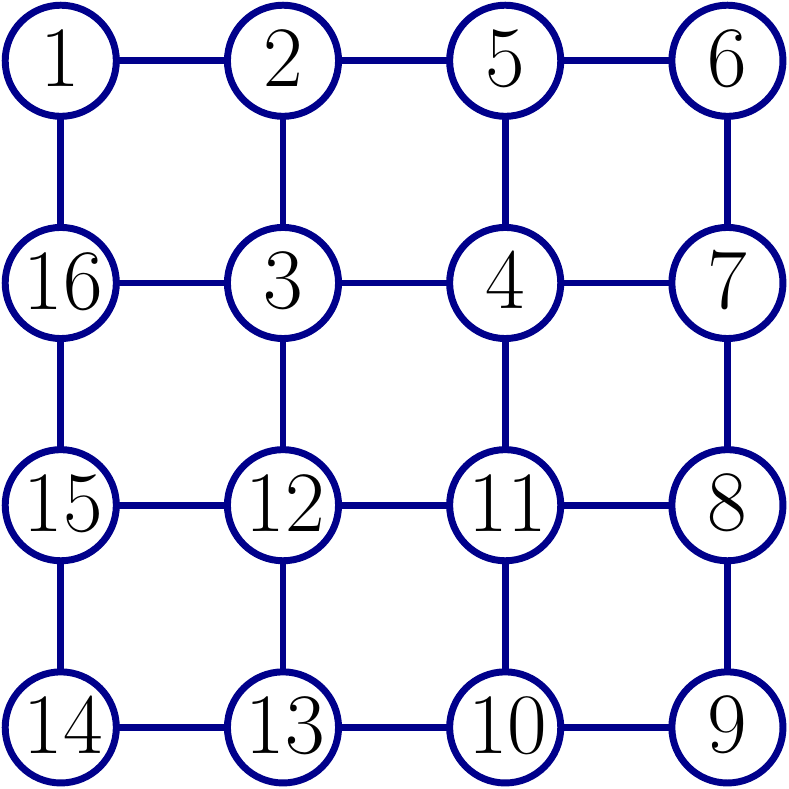}
    \caption{4 $\times$ 4 grid}
\end{subfigure}
~
\begin{subfigure}[b]{0.45\textwidth}
    \centering
    \includegraphics[scale=0.35]{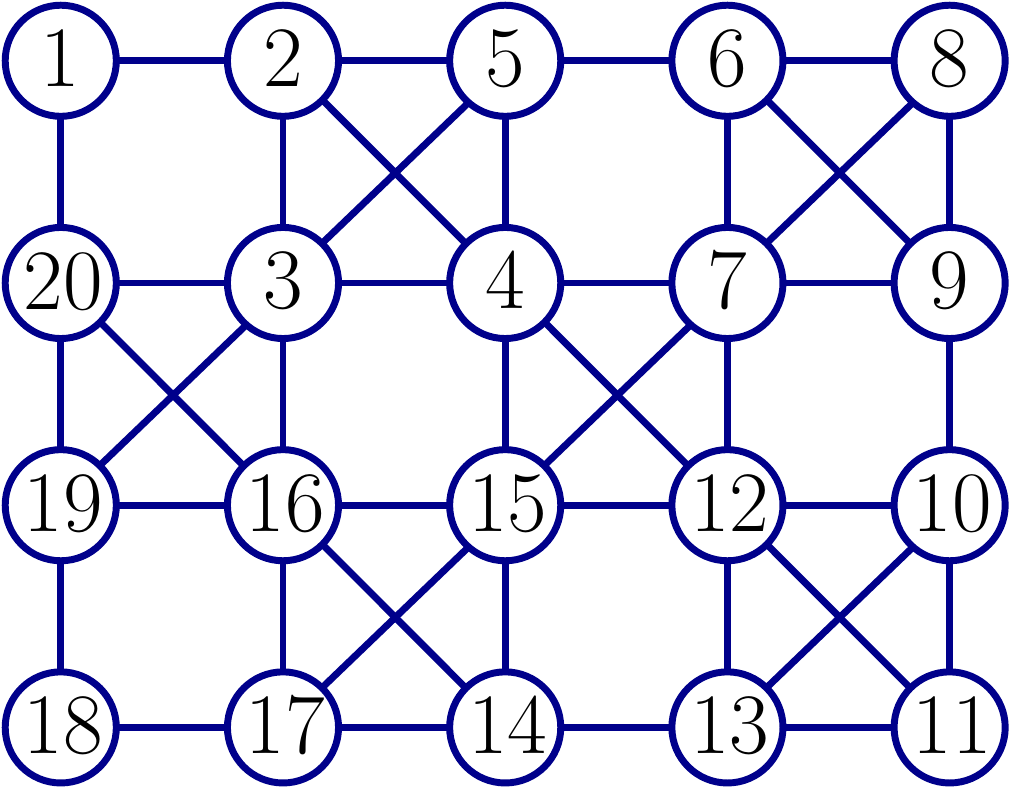}
    \caption{IBM QX20 Tokyo}
\end{subfigure}
\caption{Examples of qubit connectivity graphs from existing architectures (taken from \cite{de2021decoding}).}
\label{fig::connectivities}
\end{figure}

\subsection{Kutin \textit{et al.}'s algorithm}

In \cite{DBLP:journals/cjtcs/KutinMS07}, an algorithm for synthesizing any $n$-qubit CNOT circuit in depth at most $5n$ for the LNN architecture has been proposed. We propose a slightly reformulated version of this algorithm. The algorithm consists in two parts: 
\begin{itemize}
    \item first, reduce the operator $A$ to a north-west triangular operator $B$, i.e, such that $B[i,j] = 0$ if $i+j > n-1$ (the indices start at $0$).
    \item secondly, reduce $B$ to the identity operator. 
\end{itemize}

In both steps, each qubit/row of the matrix will be given a label and the algorithm will consist in sorting the labels while maintaining some invariants. When the labels are sorted, the invariants impose the expected result: in the first step $A$ will be north-west triangular, in the second step $B$ will be the identity operator. 

First, write 
\[ A = UPL \]

where $U$ is upper triangular, $L$ is lower triangular and $P$ is a permutation matrix. Such decomposition is always possible with a variant of the Gaussian elimination algorithm whose pseudo-code is given in Algorithm~\ref{PseudoCode_LU}. Let $J_n$ be the exchange matrix of size $n$, write 
\[ A = UPJ_nJ_nL = U  P'  W\]
where $P' = PJ_n$ and $W$ is a north-west triangular matrix. One can check that north-west triangular matrices are stable by left-multiplication with upper triangular matrices. 
Therefore, reducing $A$ to a north-west triangular matrix is equivalent to reducing an upper triangular matrix with columns permuted into an upper triangular matrix. In other words we want to do 
\[ U  P' \to U' \]
for some upper-triangular matrix $U'$, during the first step of the algorithm and 
\[ U'W  \to I \]
during the second step ($UW$ being north-west triangular).

\subsubsection{From upper with columns permuted to upper triangular}
Given $U$ and $P$, we label the rows of $UP$ as the following: row $i$ has label $j$ if $P[i,j]=1$. 
Then one can check that matrix $A=UP$ satisfies the following property that will be our invariant throughout the synthesis: 
\\ \ \\ \ \\
\fbox{\begin{minipage}{35em}
\textbf{Invariant n\textdegree1:} for each row $i$ of $A$ with label $k$, for each row $j>i$, $A[j,k]=0$. \hspace*{2.4cm} Also, we always have $A[i,k]=1$.
\end{minipage}} 
\\ \ \\ \ \\

\begin{figure}[h]
    \[
        \left(
            \begin{array}{ccccc}
                1 & 1 & \rc{col-2-r} 1 & 1 & 0\\
                1 & 1 & \gc{col-2} 0 & \rc{col-3-r} 1 & 0\\
                1 & \rc{col-1-r} 1 & 0 & \gc{col-3} 0 & 1\\
                1 & \gc{col-1} 0 & 0 & 0 & \rc{col-4-r} 1\\
                \rc{col-0-r} 1 & \egc{col-1} 0 & \egc{col-2} 0 & \egc{col-3} 0 & \gc{col-4} \egc{col-4} 0
            \end{array}
        \right)\begin{array}{c}
            3\\4\\2\\5\\1
        \end{array}
    \]
    \caption{A permuted upper-triangular matrix and its row labeling as specified by \textbf{Invariant n\textdegree1}. The first property of the invariant is depicted in green, while the second is depicted in red.}\label{fig:invariant_1}
\end{figure}

Notice that if the labels are sorted, i.e row $i$ has label $i$, and the invariant holds, then the matrix is upper-triangular (i.e. $P$ is trivial). 
Figure \ref{fig:invariant_1} depicts a matrix $A=UP$ and its labeling.

Let $A$ be the operator we manipulate during the synthesis with elementary row operations. Initially $A = UP$. We now show that given two adjacent qubits $i, i+1$, 
one can always apply a linear reversible operator between the two qubits such that the two labels can be swapped while maintaining the invariant on $A$. 

First, note that both rows verify $A[i,k] = A[i+1,k] = 0$ for any label $k$ of the rows $1 \hdots i-1$. 
It is clear that this property will remain true after any linear reversible operation between those two rows. 
The values $A[i,k], A[i+1,k]$ for any label $k$ of the rows $k > i+1$ are arbitrary and do not participate in the truth value of the invariant. 
So now it should be clear that to maintain the invariant we only need to focus on the $2 \times 2$ submatrix $B = A[[i,i+1], [k,k']]$ where $k$ is the label of row $i$ and $k'$ the label of row $i+1$. 
Because of the invariant $B$ can only have two values: 
\[ B = \begin{bmatrix} 1 & 1 \\ 0 & 1 \end{bmatrix} \text{ or } B = \begin{bmatrix} 1 & 0 \\ 0 & 1 \end{bmatrix}. \]

To swap labels $k$ and $k'$, we have to apply a 2-qubit linear reversible operator $C$ such that 
\[ CB = \begin{bmatrix} \star & 1 \\ 1 & 0 \end{bmatrix} \]
where $\star$ can be either $0$ or $1$. With $CB$ of this shape, we can assign label $k$ to row $i+1$ and label $k'$ to row $i$ while maintaining $A[i,k'] = 1, A[i+1,k]=1$ and $A[i+1,k'] = 0$. 

The choice of $C$ is simple: 
\begin{itemize} 
    \item if $B = \begin{bmatrix} 1 & 1 \\ 0 & 1 \end{bmatrix}$ then $C = E_{i,i+1}$ works, one CNOT is sufficient, 
    \item if $B = \begin{bmatrix} 1 & 0 \\ 0 & 1 \end{bmatrix}$ then $C = E_{i,i+1}E_{i+1,i}$ works, and two CNOT are sufficient. 
\end{itemize}

To summarize, given a matrix $A$ with labels on the rows and that satisfies invariant n\textdegree 1, we have a procedure that modifies $A$ such that we can swap the labels of two adjacent rows while maintaining the invariant. This procedure, that we call a \textit{box} (to follow the terminology of \cite{DBLP:journals/cjtcs/KutinMS07}) requires at most $2$ CNOT gates. Once the labels are sorted, the invariant guarantees us that $A$ will be upper triangular. What remains to do is to propose a quantum circuit made of boxes that guarantees that the labels are sorted after its execution. We want this circuit to be the shallowest possible. What is proposed in \cite{DBLP:journals/cjtcs/KutinMS07} is a LNN sorting network, an example on $7$ bits is given in Figure~\ref{sorting_network}. 
The circuit, as a box-based circuit, is of depth at most $n$. Therefore, given that each box is of depth at most $2$, the CNOT-based circuit is of depth at most $2n$. 

\begin{figure}
\center
\includegraphics[scale=1]{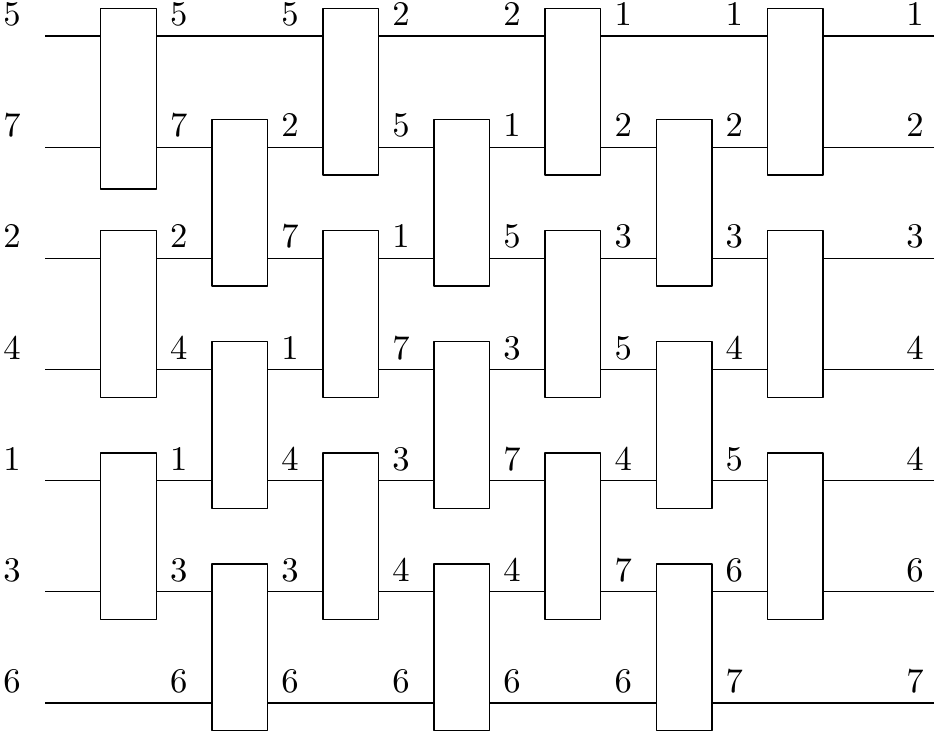}
\caption{Sorting network on $7$ bits adapted for an LNN architecture. Each box is a conditional SWAP. The network, as a box-based circuit, is of depth $n$ to sort $n$ values.}
\label{sorting_network}
\end{figure}

\subsubsection{From northwest triangular to the identity}

The principle is exactly the same than in the first step, except this time we have to reduce a north-west triangular matrix $A$. For this step, we use the following invariant: 
\\ \ \\ \ \\
\fbox{\begin{minipage}{35em}
\textbf{Invariant n\textdegree2:} for each row $i$ of $A$ with label $k$, for each row $j$, if $j>i$ then \hspace*{2.4cm} $A[j,k]=0$ ; if $j < i$ and row $j$ has label $k' < k$ then $A[j,k]=0$. \hspace*{2.4cm} Also, we always have $U[i,k]=1$.
\end{minipage}} 
\\ \ \\ \ \\

\begin{figure}[h]
    \[
        \left(
            \begin{array}{ccccc}
                1 & 0 & 1 & 1 & \rc{col-42-r}1\\
                0 & 0 & 0 & \rc{col-32-r}1 &\gc{col-42}  0\\
                0 & 1 & \rc{col-22-r}1 & \gc{col-32} 0 & 0\\
                0 & \rc{col-12-r}1 &\gc{col-22}  0 & 0 & 0\\
                \rc{col-02-r}1 & \gc{col-12} \egc{col-12} 0 & \egc{col-22} 0 & \tikzmarkend{col-32} 0 &\egc{col-42}  0
                \end{array}
        \right)\begin{array}{c}
            5\\4\\3\\2\\1
        \end{array}~~~~ \left(\begin{array}{ccccc}
            0 & 0 & 0 &\rc{col-43-r} 1 & \pc{p-51} 0\\
            1 & 0 & 1 & \gc{g-43}0 & \rc{col-53-r}1\\
            0 & \rc{col-23-r}1 & \pc{p-31}0 & 0 & \gc{g-53}0\\
            0 & \gc{g-23}0 & \rc{col-33-r}1 & 0 & 0\\
            \rc{col-13-r}1 & \egc{g-23}0 & \gc{g-33}\egc{g-33}0 & \egc{g-43}0 & \egc{g-53}0
            \end{array}\right)\begin{array}{c}
                4\\5\\2\\3\\1
            \end{array}
    \]
    \caption{(left) A north-west triangular matrix and its initial row labeling as specified by \textbf{Invariant n\textdegree2}. (right) Another matrix and its labeling satisfying the invariant.
    The first property of the invariant is depicted in green, the second in purple, and the last is depicted in red. }\label{fig:invariant_2}
\end{figure}

Initially, $A$ is northwest triangular and we assign to row $i$ the label $n-i+1$. One can check that the property is satisfied. Figure \ref{fig:invariant_2} depicts such a matrix $A$ and its labeling.

Given a pair of qubits $i, i+1$, we can again swap the labels with some modifications on $A$. Again, for the same reasons that in the first step, note that we only have to focus on the submatrix $B = A[[i,i+1], [k,k']]$. If $k < k'$ then necessarily $B = \begin{bmatrix} 1 & 0 \\ 0 & 1 \end{bmatrix}$ and there is nothing to be done. If $k' < k$ then $B$ can still have only two values: 
\[ B = \begin{bmatrix} 1 & 1 \\ 0 & 1 \end{bmatrix} \text{ or } B = \begin{bmatrix} 1 & 0 \\ 0 & 1 \end{bmatrix}. \]

The difference this time compared to the first step is that we impose to apply an operator $C$ such that 
\[ CB = \begin{bmatrix} 0 & 1 \\ 1 & 0 \end{bmatrix} \]
otherwise the invariant would not be preserved. So we have two possibilities for $C$:
\begin{itemize} 
    \item if $B = \begin{bmatrix} 1 & 1 \\ 0 & 1 \end{bmatrix}$ then $C = E_{i+1,i}E_{i,i+1}$ works, two CNOT are sufficient, 
    \item if $B = \begin{bmatrix} 1 & 0 \\ 0 & 1 \end{bmatrix}$ then $C = E_{i,i+1}E_{i+1,i}E_{i,i+1}$ works, and three CNOT are sufficient. 
\end{itemize}

With the same sorting network, we can sort the labels and reduce $A$ to the identity operator. As a box-based operator, the sorting network is of depth $n$. But now each box can be of depth at most $3$ and therefore the second step can be executed in depth at most $3n$. This gives one of the main result of \cite{DBLP:journals/cjtcs/KutinMS07}:

\begin{mythm}
Any $n$-qubit linear reversible operator can be executed in depth at most $5n$ on a LNN architecture.
\end{mythm}

\begin{proof}

Simply concatenate the two steps, the first step can be executed in depth at most $2n$, the second step can be executed in depth at most $3n$, hence the result. 

\end{proof}

\section{A block version of Kutin \textit{et al.}'s algorithm} \label{sec::block}

The main contribution of our article is a block version of Kutin \textit{et al.}'s algorithm. 
For our block version to work, we will need our target qubit topology to verify a simple structure.

Let $n$ be the number of qubits, we assume $n$ is a multiple of some integer $p$ which will be the block size. 
We pose $m=n/p$ the number of blocks. We write $b_1 = [1,2, \hdots, p], b_2 = [p+1, p+2, \hdots, p+p], \hdots, b_m$ the different blocks.
We further assume that each block $b_i$ induces a connected subgraph of the topology. Furthermore, the blocks are connected in line,
meaning that in each block $b_i, i < m$, there exists a qubit that is connected to a qubit in $b_{i+1}$.

\begin{figure}[h]
    \begin{tikzpicture}[baseline=(middle)]
        \coordinate (middle) at (0, 0.5);
        \foreach \n in {1, 2, 3, 4, 5, 6}{
            \draw (\n, 0) node[circle, draw, inner sep=3pt](1-\n){};
            \draw (\n, 1) node[circle, draw, inner sep=3pt](2-\n){};
            \draw (1-\n) -- (2-\n);
            \draw (\n, 1.5) node{\textcolor{blue}{$b_\n$}};
            \draw[draw, thick, color=blue] (0 - 0.3+\n*1, -0.3) rectangle (0.3+\n*1, 1 + 0.3);
        }
        \draw (1-1) -- (1-2) -- (1-3) -- (1-4) -- (1-5) -- (1-6);
        \draw (2-1) -- (2-2) -- (2-3) -- (2-4) -- (2-5) -- (2-6);
        \foreach \n in {1, ..., 5}{
            \draw[thick, blue, <->] (\n+0.3, 0.5) -- (\n+0.7, 0.5);
        }
    \end{tikzpicture}~~~~~~~~\begin{tikzpicture}[baseline=(middle)]
        \coordinate (middle) at (0, 1.5);
        \foreach \x in {0, ..., 5}{
            \foreach \y in {0, ..., 3}{
                \draw (\x, \y) node[circle, draw, inner sep=3pt](\x-\y){};
            }
        }
        \foreach \x in {0, ..., 5}{
            \draw (\x-0) -- (\x-1) -- (\x-2) -- (\x-3);
        }

        \foreach \y in {0, ..., 3}{
            \draw (0-\y)--(1-\y) -- (2-\y) -- (3-\y) -- (4-\y) -- (5-\y);
        }
        \foreach \y in {0, 1}{
            \foreach \x in {0, ..., 2}{
                \draw[draw, thick, color=blue] (0 - 0.3+2*\x, 1+0.3 + 2 * \y) rectangle (1+0.3+2*\x, 0 - 0.3+ 2 * \y);
            }
        }
        \foreach \y in {0, 1}{
            \foreach \x in {0, ..., 1}{
                \draw[thick, blue, <->] (1+2*\x+0.3, 2*\y + 0.5) -- (1+2*\x+0.7, 2*\y + 0.5);
            }
        }
        \draw[thick, blue, <->] ( 4.5, 1.3) -- ( 4.5, 1.7);
        \draw (0.5, 3.5) node{\textcolor{blue}{$b_1$}};
        \draw (2.5, 3.5) node{\textcolor{blue}{$b_2$}};
        \draw (4.5, 3.5) node{\textcolor{blue}{$b_3$}};
        \draw (4.5, -0.5) node{\textcolor{blue}{$b_4$}};
        \draw (2.5, -0.5) node{\textcolor{blue}{$b_5$}};
        \draw (0.5, -0.5) node{\textcolor{blue}{$b_6$}};
    \end{tikzpicture}
    \caption{Two architectures with corresponding possible block layouts. On the left: a ladder with blocks of size $p=2$. On the right: a grid topology with blocks of size $p=2\times 2= 4$. Notice the blue edges that indicates how the blocks will interact during the synthesis.}\label{fig:block_layouts}
\end{figure}
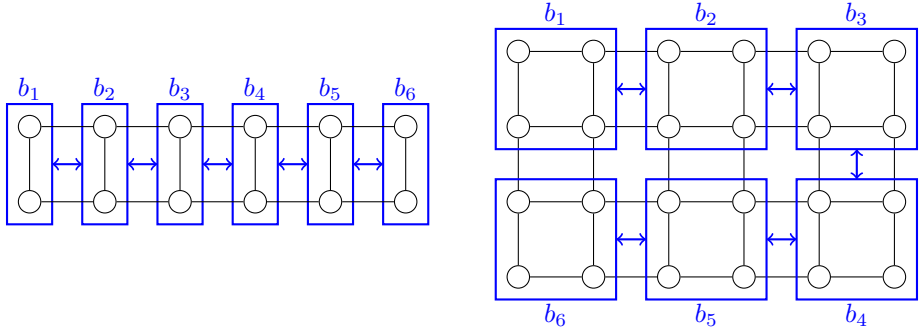





\subsection{Description of the algorithm}

Given the description we made of Kutin \textit{et al.}'s algorithm, the extension to a block version is straightforward. The idea is to perform Kutin \textit{et al.}'s algorithm on the line of blocks. 

We use the notation $A[b_i,:]$ for the natural indexing of the rows $\llbracket (i-1)*p+1, i*p \rrbracket$ of $A$. Similar notations are used to index the columns as well.

The general procedure is the same as in the case of the path graph: we start from $A = UPW$, for an upper-triangular matrix $U$ and a north-west triangular matrix $W$. 
We label the rows of $UP$ similarly according to $P$ and we reduce $UP$ to a block upper triangular matrix $U'$. Then we reduce $U'W$, a block north-west triangular matrix, to the identity. 
During the process we will sort the labels by block while maintaining two invariants.

\subsubsection{From $UP$ to block upper triangular}

During this first step, the invariant we maintain is the following: 
\\ \ \\ \ \\
\fbox{\begin{minipage}{35em}
\textbf{Invariant n\textdegree3:} for each block $b_i$ of $A$ with labels $k=[k_1, \hdots, k_p]$, for each block \hspace*{2.4cm} $b_j, j>i, U[b_j,k]=0$. Also, we have $A[b_i,k]$ invertible.
\end{minipage}} 
\\ \ \\ \ \\

    \begin{figure}[h]
        \[
            \begin{array}{c}
                2 \\ 4\\ 5\\ 1\\3\\6
            \end{array}\left(
            \begin{array}{cccccc}
                0 & \rc{rc1} 1 & 1 & \rc{rc2} 1 & 1 & 1\\
                1 & \rc{rc3} 0 & 1 & \rc{rc4} 1 & 1 & 0\\
                0 & \gc{gc1-i3} 0 & 1 & \gc{gc2-i3} 0 & 1 & 0\\
                1 & 0 & 1 & 0 & 0 & 1\\
                0 & 0 & 1 & 0 & 0 & 0\\
                0 & \egc{gc1-i3} 0 & 0 & \egc{gc2-i3}0 & 0 & 1
            \end{array}
            \right)~~~~\left(
                \begin{array}{cccccc}
                    0 & 1 & 1 & 1 & 1 & 1\\
                    1 & 0 & 1 & 1 & 1 & 0\\
                    \rc{rc1-2} 0 & 0 & 1 & 0 & \rc{rc2-2} 1 & 0\\
                    \rc{rc3-2} 1 & 0 & 1 & 0 & \rc{rc4-2} 0 & 1\\
                    \gc{gc3-i3} 0 & 0 & 1 & 0 & \gc{gc4-i3} 0 & 0\\
                    \egc{gc3-i3}0 & 0 & 0 & 0 & \egc{gc4-i3}0 & 1
                \end{array}
                \right)
        \]
        \caption{Depiction of \textbf{invariant n\textdegree3}. The matrices corresponding to the red elements need to be invertible, while green elements should all be $0$. (left) The two conditions of the invariant depicted for the first block $b_1=[1, 2]$. (right) invariants for the second block $b_2=[3, 4]$.}\label{fig:invariant_3}
    \end{figure}

This new invariant is illustrated in Figure~\ref{fig:invariant_3}. Similarly to \textbf{invariant n\textdegree1}, this invariant holds for any $UP$ with $U$ upper triangular and $P$ a permutation matrix by giving label $j$ to row $i$ if $P[i,j]=1$. 
Moreover, if the labels are block-sorted, i.e each row of $b_i$ has a label in $b_i$, then this specifies a block upper triangular matrix. 

Given two adjacent blocks $b_i, b_{i+1}$ with labels $k, k'$, we show how to assign any $p$-sized subset $k''$ of $k\cup k'$ to block $b_i$ while maintaining the invariant. Let $k''' = k\cup k' \setminus k''$. Similarly to the LNN case, we only need to focus on a particular submatrix, namely a $2p \times 2p$ matrix $B = A[[b_i, b_{i+1}], [k, k']]$. Let's write 
\[ B = \begin{bmatrix} A_1 & A_3 \\ 0 & A_2 \end{bmatrix} \]
where $A_1, A_2, A_3$ are $p \times p$, by \textbf{invariant n\textdegree2} $A_1$ and $A_2$ are invertible, $B$ is full rank, and the lower-left $p\times p$ block is $0$. If we want to assign labels $k''$ to block $b_i$ we are more interested in a column-permuted version of $B$, namely 
\[ B[:,[k'',k''']] = \begin{bmatrix} B_1 & B_2 \\ B_3 & B_4 \end{bmatrix}, \] 
where $B_1, B_2, B_3, B_4$ are arbitrary as long as $B$ is full rank. If one finds a suitable linear reversible operator $C$ such that 
\[ CB[:,[k'',k''']] = \begin{bmatrix} B_5 & B_6 \\ 0 & B_7 \end{bmatrix} \] 
then one can assign labels $k''$ to block $b_i$. $B_6$ is arbitrary. $CB$ is full rank so we have $B_5$ and $B_7$ invertible and the invariant is fully maintained. Any choice for $k''$ works but in our case we want $k''$ to be the $p$ smallest labels among $k,k'$ to perform a sorting. Then using an LNN sorting network on the blocks one can finally block sort the labels and transform $A$ into a block northwest triangular matrix.

Note that we do not need in fact to focus on the columns labeled by $k'''$ of $B$. Our goal was to show that the suitable subblocks of $B$ are still invertible, but in practice we only need to focus on 
\[ A[[b_i, b_{i+1}], k''] = \begin{bmatrix} A_1 \\ A_2 \end{bmatrix} \] 
Because the labels $k''$ can somehow be arbitrary, we cannot conclude anything on $A_1$ and $A_2$ except that they form a matrix of rank $p$. Our goal is to find a $C$ such that 
\[ CA[[b_i, b_{i+1}], k''] = \begin{bmatrix} A_3 \\ 0 \end{bmatrix} \]
with $A_3$ necessarily invertible.

\subsubsection{From block northwest triangular to the identity}

For this second step, each block $b_{n-i+1}$ has initially the labels 
\[ k_i = \llbracket (i-1)\cdot p+1, i\cdot p \rrbracket. \] 
During the second step, the labels will be block sorted but labels within each block label will not be modified. 

The invariant we maintain is the following: 
\\ \ \\ \ \\
\fbox{\begin{minipage}{35em}
\textbf{Invariant n\textdegree4:} for each block $b_i$ of $U$ with labels $k_l$, for each block \hspace*{2.4cm} $b_j, j>i, A[b_j,k_l]=0$. For each block $b_j, j < i$, with label \hspace*{2.4cm}  $k_h$, if $h < l$, then $A[b_j, k_l]=0$. Also, we have $A[b_i,k_l]$ invertible.
\end{minipage}} 
\\ \ \\ \ \\

\begin{figure}[h]
    \[
        \begin{tikzpicture}[baseline=(middle)]
            \coordinate (middle) at (0, 2*0.8-0.15);
            \foreach \n/\y in {1/1, 2/2, 3/3}{
                \draw (0, 0.8*\y) node {$k_\n$}; 
            }
        \end{tikzpicture}\left(
            \begin{array}{cccccc}
                0 & 1 & 1 & 1 & 1 & 0\\
                1 & 0 & 1 & 1 & 0 & 1\\
                1 & 1 & \rc{rc1-3} 1 & \rc{rc2-3} 0 & 0 & 0\\
                0 & 1 & \rc{rc3-3} 0 & \rc{rc4-3} 1 & 0 & 0\\
                1 & 1 & \gc{gc1-3} 0 & \gc{gc2-3} 0 & 0 & 0\\
                0 & 1 & \egc{gc1-3} 0 & \egc{gc2-3}0 & 0 & 0
                \end{array}
        \right)~~~~~~\begin{tikzpicture}[baseline=(middle)]
            \coordinate (middle) at (0, 2*0.8-0.15);
            \foreach \n/\y in {2/1, 1/2, 3/3}{
                \draw (0, 0.8*\y) node {$k_\n$}; 
            }
        \end{tikzpicture}\left(
        \begin{array}{cccccc}
            0 & 1 & 1 & 1 & 1 & 0\\
            1 & 0 & 1 & 1 & 0 & 1\\
            \rc{rc9-3}1 &\rc{rc10-3} 1 & \pc{pc1-3} 0 & \pc{pc2-3} 0 & 0 & 0\\
            \rc{rc11-3}0 &\rc{rc12-3} 1 & \pc{pc3-3} 0 & \pc{pc4-3} 0 & 0 & 0\\
            \gc{gc3-3}0 &\gc{gc4-3} 0 & \rc{rc5-3}1 &\rc{rc6-3} 0 & 0 & 0\\
            \egc{gc3-3}0 &\egc{gc4-3} 0 & \rc{rc7-3}0 &\rc{rc8-3} 1 & 0 & 0
        \end{array}
        \right)
    \]
    \caption{Depiction of \textbf{invariant n\textdegree4}. First property is depicted in green, second in purple, third in red. (left) depiction of the invariant for $i=2$ for a block north-west triangular matrix. (right) depiction of the invariant for $i=2, 3$.} \label{fig:invariant_4}
\end{figure}

Again, one can check that this property is satisfied for a block north-west triangular matrix with the chosen labeling. This is illustrated in Figure~\ref{fig:invariant_4}. Once the labels are block sorted, then the invariant specifies a block-diagonal matrix. 

Given two adjacent blocks $b_i, b_{i+1}$ with labels $k_j, k_{j'}, j > j'$, we show how to swap the blocks of labels while maintaining the invariant. We focus on the $2p \times 2p$ submatrix $B = A[[b_i, b_{i+1}], [k_{j'}, k_{j}]]$. We write 
\[ B = \begin{bmatrix} A_1 & A_3 \\ A_2 & 0 \end{bmatrix} \]
where $A_2, A_3$ are invertible. We want to find a linear reversible operator $C$ such that 
\[ CB = \begin{bmatrix} A_4 & 0 \\ 0 & A_5 \end{bmatrix}. \]

After applying $C$ the two blocks of labels can be swapped, obviously $A_4$ and $A_5$ are invertible and the invariant is preserved. Again we use an LNN sorting network on the blocks such that we eventually reduce the northwest triangular matrix into a block diagonal matrix. 

\subsubsection{From block diagonal to the identity}

This final step is straightforward: we assume we can perform a direct synthesis on each block in parallel that reduces each block to the identity gate. This can be done in depth $O(p)$, so provided the block size is a $O(1)$ this step can be done in depth $O(1)$.

\subsubsection{Summary}

Given an hardware composed of $m$ blocks of $p$ qubits laid out as a line, we have shown that we can perform Kutin \textit{et al.}'s algorithm directly on the blocks. The resulting quantum circuit consists of two sorting networks made of "block-boxes", i.e, linear reversible operators that act on 2 adjacent blocks ($2p$ adjacent qubits). During the first step, those block-boxes transform binary matrices of the form 
\[ \begin{bmatrix} A_1 \\ A_2 \end{bmatrix} \]
to 
\[ \begin{bmatrix} A_3 \\ 0 \end{bmatrix} \]
where $A_3$ is invertible, $A_1$ and $A_2$ are arbitrary. This is our Problem 1: \\

\boxed{\begin{tabular}{>{\bfseries}ll}
  Input    & $\circ$ an integer $p > 0$, \\
           & $\circ$ A full rank $2p \times p$ boolean matrix $B = \begin{bmatrix} B_1 \\ B_2 \end{bmatrix}$, \\
           & $\circ$ a connectivity graph $G$ of size $2p$ giving the available row operations on $B$, \\
& \\
  Problem 1  & Find a sequence of row operations $C$ $G$-compliant such that:  \\
             & $\circ$ $CB = \begin{bmatrix} B_3 \\ 0 \end{bmatrix}$.
\end{tabular}}
\
\\
During the second step, the block-boxes transform binary matrices of the form 
\[ \begin{bmatrix} A_1 & A_3 \\ A_2 & 0 \end{bmatrix} \]
to 
\[ \begin{bmatrix} A_4 & 0 \\ 0 & A_5 \end{bmatrix} \]
where $A_2, A_3, A_4$ and $A_5$ are invertible. $A_1$ is arbitrary. This is our Problem 2: \\ \\
\boxed{\begin{tabular}{>{\bfseries}ll}
  Input    & $\circ$ an integer $p > 0$, \\
           & $\circ$ A $2p \times 2p$ boolean matrix $B = \begin{bmatrix} B_1 & B_3 \\ B_2 & 0 \end{bmatrix}$ with $B_2$ and $B_3$ invertible, \\
           & $\circ$ a connectivity graph $G$ of size $2p$ giving the available row operations on $B$, \\
& \\
  Problem 2  & Find a sequence of row operations $C$ $G$-compliant such that:  \\
             & $\circ$ $CB = \begin{bmatrix} B_4 & 0 \\ 0 & B_5 \end{bmatrix}$.
\end{tabular}}
\ \\
A third step simply consists on the synthesis of linear reversible operators on $p$ qubits. This is our Problem 3.

A pseudo-code is given in Algorithm~\ref{PseudoCode}. A detailed example on $8$ qubits with blocks of size $2$ and a topology of a ladder with diagonals is given in Figure~\ref{example}. For the moment our framework is very generic because we have not proposed algorithms to synthesize shallow block-boxes. This corresponds to the function "ZeroBlock" in our pseudo-code. This will be the subject of the next section. 

Let's now derive a general formula for the total depth. If we write $d_1(p)$, resp. $d_2(p)$, the maximum depth needed to perform the transformations necessary during the first step, resp. the second step, and if we note $d^{*}(p)$ the maximum depth required to do the synthesis of a linear reversible operator on $p$ qubits, then 
\[ d(n) \leq \frac{n}{p} \times \bigg(d_1(p) + d_2(p) \bigg) + d^*(p). \]

The main question now is to propose algorithms for different architectures that will give interesting values for $d_1, d_2, d^*$.

\begin{algorithm}
\caption{A block version of Kutin \textit{et al.}'s algorithm \cite{DBLP:journals/cjtcs/KutinMS07} for a depth-optimized synthesis of CNOT circuits with hardware constraints.}
\label{PseudoCode}
\begin{algorithmic} 
\REQUIRE $n > 0, \; 0 < p < n, \; \; A \in \mathbb{F}_2^{n \times n}$ \tcp{$p$ is the block size}
\ENSURE $C$ is a CNOT-circuit implementing $A$ compliant with the hardware topology
\STATE
\Fn{Synthesis($A$, $n$, $p$, topology)}
{
    \texttt{\\}
    \tcp{Step 1}
    \STATE $U$, labels $\leftarrow$ UPL($A, n$)
    \STATE $C \leftarrow$ SortLabels($U$, labels, $n$, $p$, topology, step=$1$)
    \texttt{\\}
    \texttt{\\}
    \tcp{Step 2}
    \tcp{$::$ means concenation}
    \STATE $NW \leftarrow$ ApplyCircuit($A, C$)
    \STATE labels $\leftarrow [ n-i-1 \text{ for } i = 0 \hdots n-1 ]$
    \STATE $C \leftarrow C :: $ SortLabels($NW$, labels, $p$, topology, step=$2$) 
    \texttt{\\}
    \texttt{\\}
    \tcp{Step 3}
    \tcp{Shift shifts the qubits on which a circuit is applied}
    \FOR{$i = 0 \hdots n//p$}
    \STATE $C \leftarrow C:: $ Shift(DirectSynthesis($NW$[ip:(i+1)p], topo), ip)
    \ENDFOR
    \texttt{\\}
    \texttt{\\}
    \RETURN reverse(C)
}

\Fn{SortLabels($U$, labels, $n$, $p$, topology, step)}
{
    \texttt{\\}
    \STATE \#blocks $\leftarrow n//p$ 
    \STATE $C \leftarrow []$
    \STATE shift $\leftarrow 0$
    \FOR{$i = 0 \hdots$ \#blocks$-1$}
    \FOR{$j = 0 \hdots$ (\#blocks-shift)//2}
    \STATE block1 $\leftarrow \llbracket 2pj + p\times \text{shift}, \; \; \; \; \; \, \, \; 2pj + p\times \text{shift}+p-1 \rrbracket$
    \STATE block2 $\leftarrow \llbracket 2pj + p\times \text{shift}+p, \; 2pj + p\times \text{shift} + 2p-1 \rrbracket$
    \STATE rows $\leftarrow [\text{block1, block2}]$
    \STATE SortTwoBlockLabels($U$, $p$, labels, rows, topology, step=step, $C$)
    \ENDFOR
    \STATE shift $\leftarrow$ (shift $ + 1) \mod 2$ 
    \ENDFOR
    \texttt{\\}
    \texttt{\\}
    \RETURN C
}

\Fn{SortTwoBlockLabels($U$, $p$, labels, rows, topology, step, C)}
{
    \texttt{\\}
    \STATE sorted\_labels $\leftarrow$ sort(labels[rows])
    \STATE $V \leftarrow U[rows, \text{sorted\_labels}]$
    \STATE local\_C $\leftarrow$ ZeroBlock($V$, topology, step=step)
    \FOR{gate in local\_C}
        \STATE $C \leftarrow C:: $ CNOT(rows[control(gate)], rows[target(gate)])
        \STATE $U$[rows[target(gate)], :] $\leftarrow U \oplus U$[rows[control(gate)],:]
    \ENDFOR
    \texttt{\\}
    \STATE labels[rows] $\leftarrow$ sorted\_labels
}

\end{algorithmic}
\end{algorithm}

\begin{algorithm}
\caption{Tuned $LU$ decomposition.}
\label{PseudoCode_LU}
\begin{algorithmic} 
\REQUIRE $n > 0, \; \; \; A \in \mathbb{F}_2^{n \times n}$
\ENSURE $U$ is an upper triangular operator with columns permuted such that $U^{-1}A$ is north-west triangular
\STATE
\Fn{UPL($A$, $n$)}
{
    \texttt{\\}
    \STATE $U \leftarrow I_n$
    \FOR{$i = n-1 \hdots 0$}
        \STATE $pivot = n-1$
        \IF{$A[:,i] != 0$}
            \WHILE{$A[pivot,i]=0$}
                \STATE $pivot \leftarrow pivot-1$
            \ENDWHILE
            \texttt{\\}
            \FOR{$j = 0 \hdots n-1$}
                \IF{$j != pivot$ and $A[j,i]=1$}
                    \STATE $A[j,:] \leftarrow A[j,:] \oplus A[pivot,:]$
                    \STATE $U[:,pivot] \leftarrow U[:,pivot] \oplus U[:,j]$
                \ENDIF
            \ENDFOR
            \FOR{$j = i-1 \hdots 0$}
                \IF{$A[pivot,j]=1$}
                    \STATE $A[:,j] \leftarrow A[:,j] \oplus A[:,i]$
                \ENDIF
            \ENDFOR
        \ENDIF
    \ENDFOR
    \texttt{\\}
    \texttt{\\}
    \STATE labels $\leftarrow [0]*n $
    \FOR{$i=0 \hdots n-1$}
        \STATE $j \leftarrow 0$
        \WHILE{$A[i,j]=0$}
            \STATE $j \leftarrow j+1$
        \ENDWHILE
        \STATE labels[i] $\leftarrow n-j-1$
    \ENDFOR
    \texttt{\\}
    \texttt{\\}
    \STATE inv\_labels $\leftarrow [0]*n$
    \FOR{$i = 0 \hdots n-1$}
        \STATE inv\_labels[labels[$i$]] $\leftarrow i$
    \ENDFOR
    \texttt{\\}
    \texttt{\\}
    \STATE $U \leftarrow U[:,inv\_labels]$
    \texttt{\\}
    \texttt{\\}
    \RETURN $U$, labels
}
\end{algorithmic}
\end{algorithm}

\begin{figure}
\vspace*{-2cm} 
\begin{subfigure}[b]{\textwidth}
\hspace*{-3cm} \includegraphics[scale=0.70]{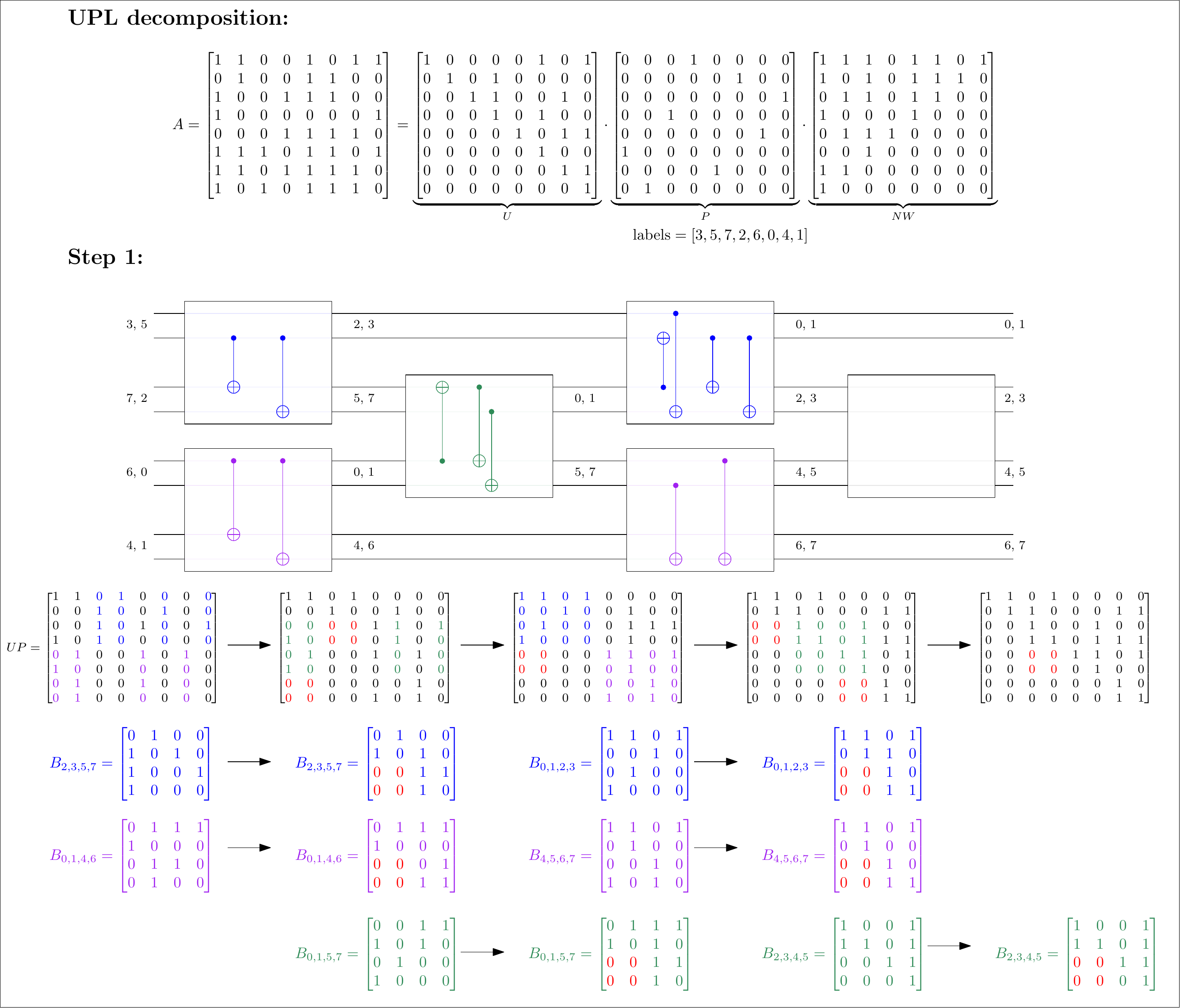}
\caption{UPL decomposition and Step 1. Total depth for Step 1 $= 7$.}
\end{subfigure}

\caption{Illustration of our block algorithm with an example on $8$ qubits and blocks of size $2$ and full qubit connectivity between two blocks. First we present the sorting network, what CNOT gates are applied and how the sorting is progressively done. Then we show directly on the matrix what row operations are done and what submatrices we are considering to determine which operator to apply. Note that we do not necessarily choose the shallowest operators to perform step $1$ or $2$. The final circuit has depth $25$ (in fact $24$ because one CNOT from one box can be merged to another box) while the execution of $A$ on a LNN architecture would require a circuit of depth $38$.}
\label{example}
\end{figure}

\begin{figure}
\ContinuedFloat
\vspace*{-4cm} 
\begin{subfigure}[b]{\textwidth}
\hspace*{-3cm} \includegraphics[scale=0.70]{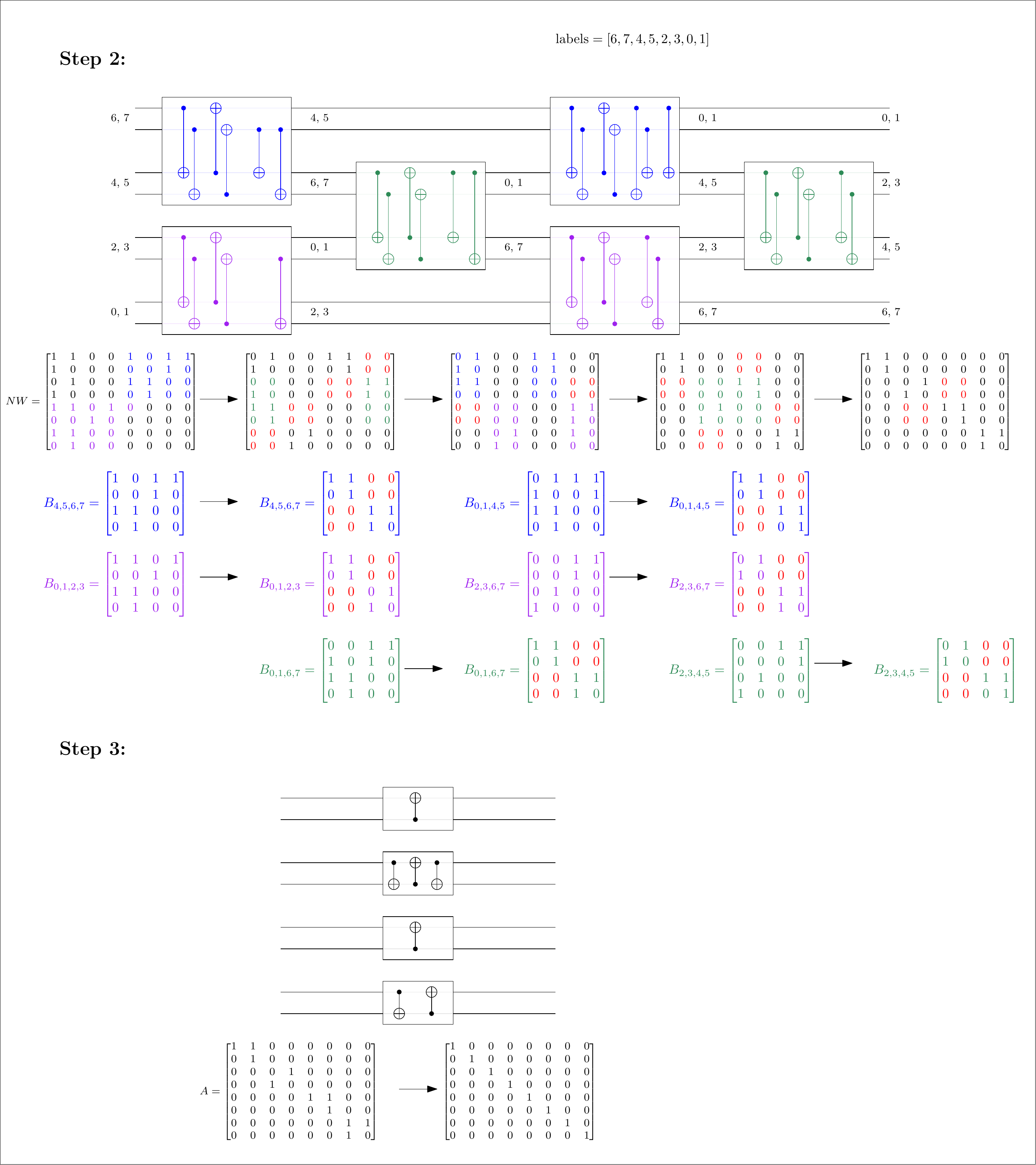}
\caption{Step 2 and Step 3. Total depth for Steps 2 and 3 $= 15 + 3 = 18$.}
\end{subfigure}

\caption{Illustration of our block algorithm with an example on $8$ qubits and blocks of size $2$ and full qubit connectivity between two blocks. First we present the sorting network, what CNOT gates are applied and how the sorting is progressively done. Then we show directly on the matrix what row operations are done and what submatrices we are considering to determine which operator to apply. Note that we do not necessarily choose the shallowest operators to perform step $1$ or $2$. The final circuit has depth $25$ (in fact $24$ because one CNOT from one box can be merged to another box) while the execution of $A$ on a LNN architecture would require a circuit of depth $38$.}
\end{figure}

\section{Practical implementations for different qubit connectivities} \label{sec::practical}

Before proposing strategies, note that our two problems are unchanged after column operations. Either during Problem 1 when working on 
\[ \begin{bmatrix} A_1 \\ A_2 \end{bmatrix} \]
or during Problem 2 with
\[ \begin{bmatrix} A_1 & A_3 \\ A_2 & 0 \end{bmatrix} \]
any column operation on $A_1, A_2$ during Problems 1 and 2 gives another problem whose solutions are exactly the same. One can check that undoing the column operations on 
\[ \begin{bmatrix} A_3 \\ 0 \end{bmatrix} \]
resp.,
\[ \begin{bmatrix} A_4 & 0 \\ 0 & A_5 \end{bmatrix} \]
will not modify the desired structure.

\subsection{Exhaustive search for small blocks}

For small $p$, it might be interesting to try an exhaustive search in order to have the best possible bounds on the depth. The search will consist in a breadth-first search. To perform such search, we need three elements: 
\begin{itemize}
    \item starting roots, these elements need depth $0$ to solve Problem 1 or 2,
    \item a set of all available operations of depth $1$,
    \item a characterization of the set of all elements we need to find.
\end{itemize}

The matrix at a maximum distance from the roots will give the maximum depth required to solve Problem 1 or 2.
The roots for Problem 1 are all matrices of the form 
\[ \begin{bmatrix} A_3 \\ 0 \end{bmatrix} \]
with $A_3$ invertible. Similarly for Problem 2 the roots are all matrices of the form 
\[ \begin{bmatrix} A_4 & 0 \\ 0 & A_5 \end{bmatrix} \]
with $A_4, A_5$ invertible. 

For Problem 1, we need to cover all full rank matrices of size $2p \times p$. For Problem 2, we need to cover all matrices of the form 
\[ \begin{bmatrix} A_1 & A_3 \\ A_2 & 0 \end{bmatrix} \]
with $A_2, A_3$ invertible and $A_1$ is arbitrary. 

Fortunately, we can rely on the invariance of the problems by column operations to reduce the search space. With column operations, we can put our matrices in reduced column-echelon form. Such form is unique for one given matrix. This has two consequences:
\begin{itemize}
    \item first, both Problem 1 and 2 have only one root, namely 
\[ \begin{bmatrix} I_p \\ 0 \end{bmatrix} \] 
for Problem 1 and 
\[ \begin{bmatrix} I_p & 0 \\ 0 & I_p \end{bmatrix} \]
for Problem 2,
    \item secondly, for Problem 1 we only need to cover the set of matrices of size $2p \times p$ in reduced column-echelon form. For Problem 2 we will have to search through the set of matrices of the form 
    \[ \begin{bmatrix} R_1 & R_2 \end{bmatrix} \]
    where $R_1, R_2$ are $2p \times p$ and are in reduced column-echelon form. Note that not all matrices of this form are of interest (the targeted matrices have more structure) but this is the smallest stable set in which we can do the search.
\end{itemize}

The set of available operations can be computed in the following way: 
\begin{itemize}
    \item enumerate all possible matchings of the connectivity graph $G$,
    \item each edge of a matching is a pair of qubits on which two different CNOT gates can be applied, depending on which qubit is the control and the target. For each matching do all possible combinations, i.e, if for each edge we assign a $0$ or a $1$ characterizing which CNOT is applied.
\end{itemize}

The results of this exhaustive search are given in Table~\ref{search}. 
First, note that for Problem 1 we recover the number of full rank matrices under reduced row echelon form. This is equivalent to counting the number of different $p$ dimensional subspaces of the $2p$ dimensional vector space over $GF(2)$. This is known to be the Gaussian binomial coefficient $\binom{2p}{p}_2$ (see, e.g, \cite{andrews1998theory}). 
For Problem 2, all target matrices can be put in the canonical form 

\[ \begin{bmatrix} A & I \\ I & 0 \end{bmatrix} \]

where $A$ is an arbitrary $p \times p$ boolean matrix. 
We recover the fact that there can be $2^{p^2}$ of them. From those maximal depth values, we deduce the complexities summarized in Table~\ref{summary} where we explicit the depth to solve each Problem for each architecture and the total depth required for full operator synthesis. 

Table \ref{table:summary_of_summary} gives an overview of the achieved asymptotic depth complexities for the various explored architectures.
The main results, in our opinion, are the ones for the grid: any CNOT circuit on $n$ qubits can be executed in depth $4n + O(1)$ on a grid and $15n/4 + O(1)$ when we add diagonal interactions. 

We also want to emphasize that when there is a full connectivity between blocks of size $4$, the total depth is $7n/4 + O(1)$, which is better than the $2n$ result given by the extension of Kutin at al.'s to the 
full qubit connectivity \cite{de2021reducing}. 
The best practical method so far synthesizes CNOT circuits in depth $n + o(n)$ \cite{maslov2022depth}. 
This shows that with slightly more connectivity between the qubits than the LNN architecture we are able to significantly reduce the computational depth of CNOT circuits and have close to similar results to 
the full connectivity case.

\begin{table}[h]
    \centering
    \caption{Summary of the depth complexities for the different explored architectures. (*) the enumeration never finished for step 2 of the Grid + diagonals layout, thus we used the result for the step 2 of Grid (which is a subgraph of this layout).}
    \begin{tabular}{llll}
        \toprule
        Architecture & Step 1 & Step 2 & Total depth \\
        \cmidrule(lr){1-1} \cmidrule(lr){2-2} \cmidrule(lr){3-3} \cmidrule(lr){4-4}
        Width-2 ladder              & $2n$   & $2n$      & $4n+O(1)$\\
        Width-2 ladder + diagonals  & $3n/2$ & $2n$      & $7n/2+O(1)$ \\
        Width-3 ladder              & $5n/3$ & $2n$      & $11n/3+O(1)$\\
        Width-3 ladder + diagonals  & $4n/3$ & $5n/3$    & $3n+O(1)$\\
        3-qubit all-to-all          & $n$    & $4n/3$    & $7n/3+O(1)$ \\
        Width-4 ladder              & $3n/2$ & $7n/4$    & $13n/4+O(1)$\\
        Width-4 ladder + diagonals  & $5n/4$ & $5n/4$    & $5n/2+O(1)$ \\
        Grid                        & $7n/4$ & $9n/4$    & $4n+O(1)$\\
        Grid + diagonals            & $3n/2$ & $9n/4$(*) & $15n/4 + O(1)$(*)\\
        4-qubit all-to-all          & $3n/4$ & $n$       & $7n/4+O(1)$\\
        \bottomrule
    \end{tabular}
    
\label{table:summary_of_summary}
\end{table}

\begin{table}[t]
\begin{adjustwidth}{-4cm}{}
\centering
\caption{The table gives, for each architecture, for each problem and for each depth, the number of operators in reduced form that can solve the given problem in the given depth for the given architecture. (*) The brute-force search did not terminate.}

\renewcommand{\arraystretch}{1.05}
\scalebox{0.63}{
\begin{tabular}{cccccrrrrrrrrrrcr}
\toprule
\toprule
Architecture & Block size & Local topology & Problem & \multicolumn{11}{c}{Depth} & & Total \\
\cmidrule(lr){1-1} \cmidrule(lr){2-2} \cmidrule(lr){3-3} \cmidrule(lr){4-4} \cmidrule(lr){5-15} \cmidrule(lr){17-17}
&&&& \multicolumn{1}{c}{} & \multicolumn{1}{r}{$\; \; \; d=0$} & \multicolumn{1}{r}{$\; \; \; d=1$} & \multicolumn{1}{r}{$\; \; \; d=2$} & \multicolumn{1}{r}{$\; \; \; d=3$} & \multicolumn{1}{r}{$\; \; \; d=4$} & \multicolumn{1}{r}{$\; \; \; d=5$} & \multicolumn{1}{r}{$\; \; \; d=6$} & \multicolumn{1}{r}{$\; \; \; d=7$} & \multicolumn{1}{r}{$\; \; \; d=8$} & \multicolumn{1}{r}{$\; \; \; d=9$} \\
&&&& \multicolumn{1}{c}{} & \multicolumn{1}{c}{} & \multicolumn{1}{c}{} & \multicolumn{1}{c}{} & \multicolumn{1}{c}{} & \multicolumn{1}{c}{} & \multicolumn{1}{c}{} & \multicolumn{1}{c}{} & \multicolumn{1}{c}{} &&
\\
\multirow{2}{*}{Width-2 ladder} & \multirow{2}{*}{$2$} & \multirow{2}{*}{\includegraphics[scale=0.5]{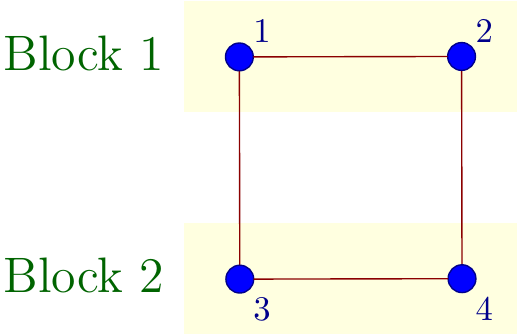}} & 1 && 1 & 3 & 14 & 15 & 2 & & & & & & & 35   \\
                                &&& 2 && 0 & 0 & 1 & 7 & 8 & & & & & & & 16 \\
&&&&&&&&&&&&&& \\ 
&&&&&&&&&&&&&& \\                                 
\\                       
\multirow{2}{*}{Width-2 ladder + diagonals} & \multirow{2}{*}{$2$} & \multirow{2}{*}{\includegraphics[scale=0.5]{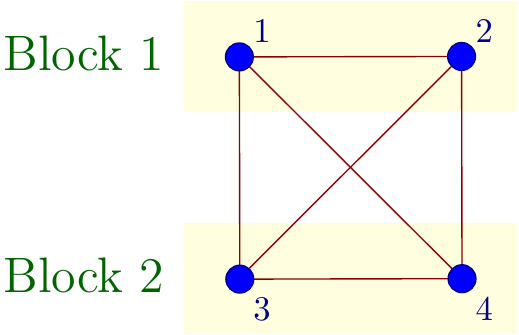}} & 1 && 1 & 6 & 19 & 9 & & & & & & & & 35   \\
                                &&& 2 && 0 & 0 & 2 & 10 & 4 & & & & & & & 16  \\
&&&&&&&&&&&&&& \\ 
&&&&&&&&&&&&&& \\                              
\\                       
\multirow{2}{*}{Width-3 ladder} & \multirow{2}{*}{$3$} & \multirow{2}{*}{\includegraphics[scale=0.5]{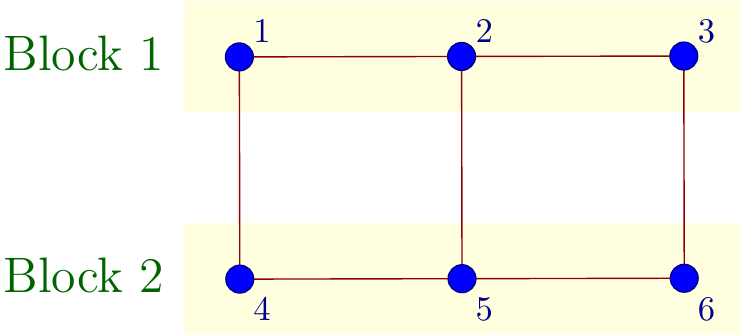}} & 1  && 1 & 7 & 91 & 538 & 736 & 22 & & & & & & 1395  \\
                                &&& 2 && 0 & 0 & 1 & 29 & 206 & 269 & 7 & & & & & 512 \\
&&&&&&&&&&&&&& \\ 
&&&&&&&&&&&&&& \\                               
\\                       
\multirow{2}{*}{Width-3 ladder + diagonals} & \multirow{2}{*}{$3$} & \multirow{2}{*}{\includegraphics[scale=0.5]{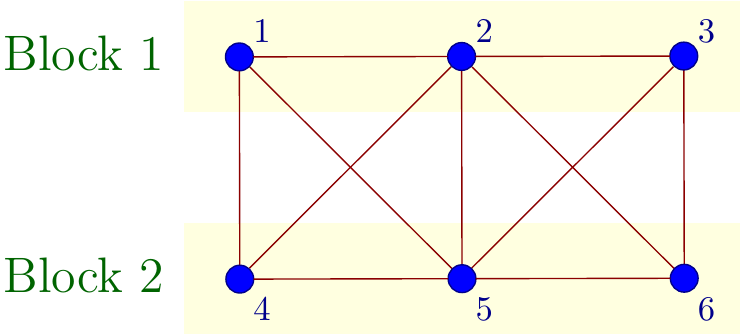}} & 1 && 1 & 21 & 293 & 1042 & 38 &&&&&&& 1395   \\
                                &&& 2 && 0 & 0 & 3 & 99 & 406 & 4 & & & & & & 512 \\
&&&&&&&&&&&&&& \\  
&&&&&&&&&&&&&& \\                              
\\                                
\multirow{2}{*}{3-qubit all-to-all} & \multirow{2}{*}{$3$} & \multirow{2}{*}{\includegraphics[scale=0.5]{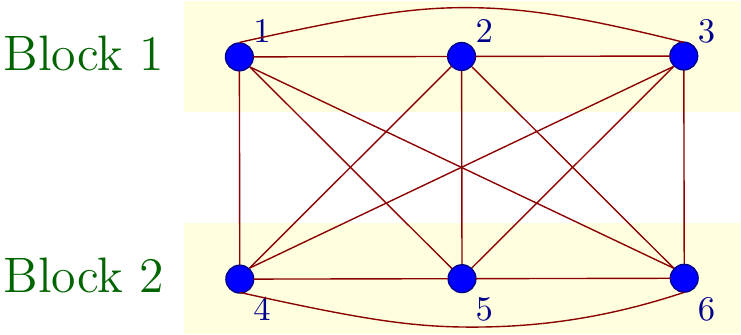}} & 1 && 1 & 33 & 649 & 712 &&&&&&&& 1395  \\
                                &&& 2 && 0 & 0 & 6 & 250 & 256 & & & & & & & 512 \\
&&&&&&&&&&&&&& \\ 
&&&&&&&&&&&&&& \\                               
\\                       
\multirow{2}{*}{Width-4 ladder} & \multirow{2}{*}{$4$} & \multirow{2}{*}{\includegraphics[scale=0.5]{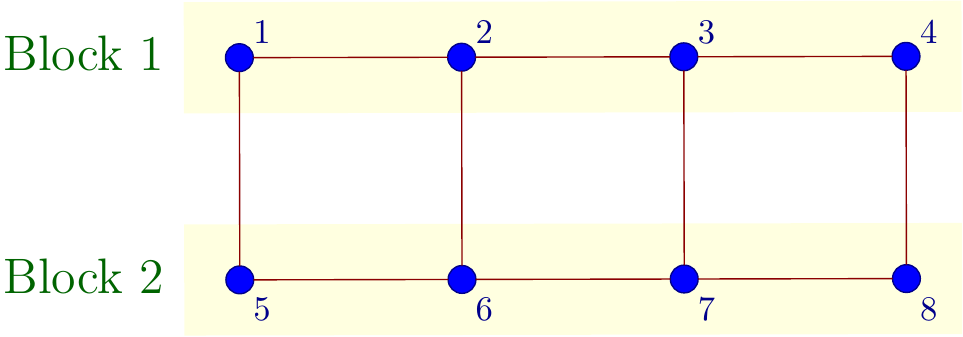}} & 1 && 1 & 15 & 543 & 9746 & 75037 & 110338 & 5107 & & & & & 200787   \\
                                &&& 2 && 0 & 0 & 1 & 117 & 2692 & 20991 & 38695 & 3040 & & & & 65536 \\
&&&&&&&&&&&&&& \\ 
&&&&&&&&&&&&&& \\                              
\\  
\multirow{2}{*}{Width-4 ladder + diagonals} & \multirow{2}{*}{$4$} & \multirow{2}{*}{\includegraphics[scale=0.5]{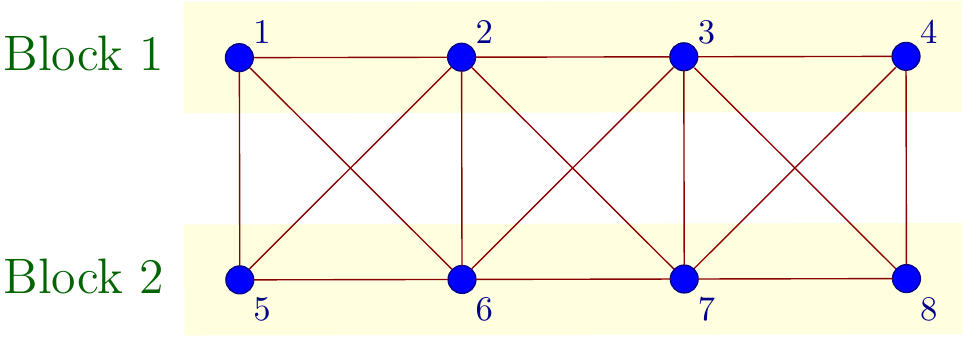}} & 1 && 1 & 70 & 3671 & 76393 & 120118 & 534 & & & & & & 200787  \\
                                &&& 2 && 0 & 0 & 5 & 772 & 21580 & 43179 & & & & & & 65536 \\
&&&&&&&&&&&&&& \\  
&&&&&&&&&&&&&& \\                              
\\  
\multirow{2}{*}{Grid} & \multirow{2}{*}{$4$} & \multirow{2}{*}{\includegraphics[scale=0.5]{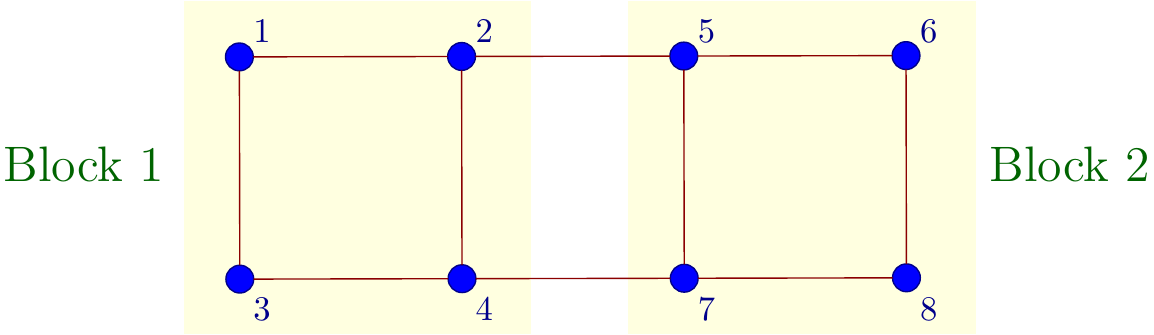}} & 1 && 1 & 3 & 57 & 1873 & 29293 & 136771 & 32733 & 56 & & & & 200787  \\
                                &&& 2 && 0 & 0 & 0 & 0 & 0 & 0 & 25 & 5263 & 55203 & 5045 & & 65536 \\
&&&&&&&&&&&&&& \\  
&&&&&&&&&&&&&& \\                              
\\    
\multirow{2}{*}{Grid + diagonals} & \multirow{2}{*}{$4$} & \multirow{2}{*}{\includegraphics[scale=0.5]{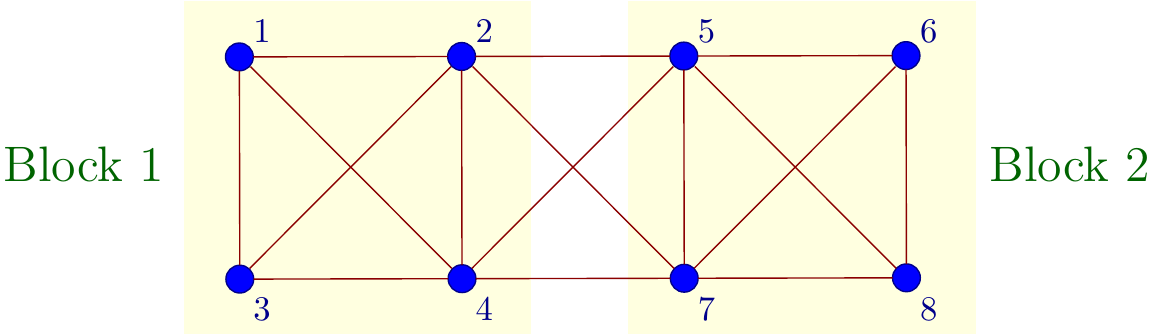}} & 1 && 1 & 6 & 275 & 11426 & 111564 & 75979 & 1536 & & & & & 200787   \\
                                &&& 2 && / & / & / & / & / & / & / & / & / & & & (*). \\
&&&&&&&&&&&&&& \\ 
&&&&&&&&&&&&&& \\                               
\\   
\multirow{2}{*}{4-qubit all-to-all} & \multirow{2}{*}{$4$} & Not shown for clarity. & 1 && 1 & 208 & 42973 & 157605 & & & & & & & & 200787   \\
                                &&& 2 && 0 & 0 & 24 & 14111 & 51401 & & & & & & & 65536 \\
&&&&&&&&&&&&&&                               
\\                       
\bottomrule
\bottomrule
\end{tabular}}

\label{search}
\end{adjustwidth}
\end{table}

\begin{table}[t]
\begin{adjustwidth}{-2cm}{}
\centering
\caption{Summary of the depth complexities for different architectures.}
\renewcommand{\arraystretch}{1.1}
\scalebox{0.7}{
\begin{tabular}{ccccccc}
\toprule
\toprule
Architecture & Block size & Local topology & Depth Problem 1 & Depth Problem 2 & Depth Problem 3 & Total depth \\
\cmidrule(lr){1-1} \cmidrule(lr){2-2} \cmidrule(lr){3-3} \cmidrule(lr){4-4} \cmidrule(lr){5-5} \cmidrule(lr){6-6} \cmidrule(lr){7-7}
\\
Width-2 ladder & $2$ & \includegraphics[scale=0.5]{images/ladder2.pdf} & $2n$ & $2n$ & $O(1)$ & $4n + O(1)$ \\
\\
Width-2 ladder + diagonals & $2$ & \includegraphics[scale=0.5]{images/ladder2_diagonal.pdf} & $3n/2$ & $2n$ & $O(1)$ & $7n/2 + O(1)$ \\
\\
Width-3 ladder & $3$ & \includegraphics[scale=0.5]{images/ladder3.pdf} & $5n/3$ & $2n$ & $O(1)$ & $11n/3 + O(1)$ \\
\\
Width-3 ladder + diagonals & $3$ & \includegraphics[scale=0.5]{images/ladder3_diagonal.pdf} & $4n/3$ & $5n/3$ & $O(1)$ & $3n + O(1)$ \\
\\
3-qubit all-to-all & $3$ & \includegraphics[scale=0.5]{images/all_to_all3.pdf} & $n$ & $4n/3$ & $O(1)$ & $7n/3 + O(1)$ \\
\\
Width-4 ladder & $4$ & \includegraphics[scale=0.5]{images/ladder4.pdf} & $3n/2$ & $7n/4$ & $O(1)$ & $13n/4 + O(1)$ \\
\\
Width-4 ladder + diagonals & $4$ & \includegraphics[scale=0.5]{images/ladder4_diagonal.pdf} & $5n/4$ & $5n/4$ & $O(1)$ & $5n/2 + O(1)$ \\
\\
Grid & $4$ & \includegraphics[scale=0.5]{images/grid.pdf} & $7n/4$ & $9n/4$ & $O(1)$ & $4n+O(1)$ \\
\\
Grid + diagonals & $4$ & \includegraphics[scale=0.5]{images/grid_diagonal.pdf} & $3n/2$ & $9n/4$ & $O(1)$ & $15n/4+O(1)$  \\
\\
4-qubit all-to-all & $4$ & Not shown for clarity. & $3n/4$ & $n$ & $O(1)$ & $7n/4+O(1)$  \\
\\
\bottomrule
\bottomrule
\end{tabular}}

\label{summary}
\end{adjustwidth}
\end{table}





\subsection{Combining block layouts}

Equipped with the asymptotic bounds for step 1 and 2 of our algorithm for different architectures, one can try to combine them in order to improve the overall depth bound in some case.
We propose two possible improvements formalized in the following propositions.

\begin{mypropo}\label{prop:ladder_2}
    In a $2\!\times\! 2L$ grid layout ($n=4L$), any $n-$qubits linear reversible operator can be implemented in depth at most $15n/4 + O(1)$.
\end{mypropo}

In this hardware setting, one can either divide the grid in $2L$ blocks of size $p=2$ and use the ``Width-2 ladder'' bounds, or divide the grid into $L$ blocks of size $4$ and use the ``Grid'' bounds.
One can combine step 1 of the grid layout ($7n/4$) with step 2 of the Width-2 ladder layout ($2n$).
In order to achieve this, we need to be able to turn a block north-west triangular matrix with block size $4$ into a block north-west triangular matrix with block size $2$.
This step is at most as costly as performing a full block synthesis, and can thus be bounded by $d^*(4)$, which is constant.
Hence, the overall depth bound is of $7n/4 + 2n + O(1) = 15n/4 + O(1)$.

\begin{figure}[h]
    \centering
    \scalebox{0.7}{
    \begin{tikzpicture}
        \newsavebox{\rightpiece}
        \savebox{\rightpiece}{
            \foreach \y in {0, ..., 2}{
                \draw[fill, color=orange!20] (-0.2, -2 * \y + 0.2) rectangle (1.2,  -2 * \y + -1.2);
                \draw[fill, color=orange!20] (2-0.2, -2 * \y +  0.2) rectangle (3.2, -2 * \y +  -1.2);
            }
            \foreach \x in {0, ...,3}{
                \foreach \y in {0, ...,5}{
                    \draw (\x, -\y) node[circle, draw, inner sep=3pt](\x-\y){};
                }
            }
            \foreach \y in {0, ...,5}{
                \draw (0-\y)--(1-\y)--(2-\y)--(3-\y);
            }
            \foreach \x in {0, ...,3}{
                \draw(\x-0) -- (\x-1) -- (\x-2) -- (\x-3) -- (\x-4) -- (\x-5);
            }
            \draw (3-0) .. controls (4,-0.33) and (4,-1.66) .. (3-2);
            \draw (3-1) .. controls (4,-1.33) and (4,-2.66) .. (3-3);
            \draw[dashed] (3-4) .. controls (4,-4.33) and (4,-5.66) .. (3.2,-5.8);
            \draw[dashed] (3-5) .. controls (4,-5.33) and (4,-6.66) .. (3.2,-6.8);
            \foreach \y in {0, ..., 2}{
                \foreach \x in {0, 1, 2, 3}{
                    \draw[color=blue, thick] (\x - 0.3, -2 * \y + 0.3) rectangle (\x+0.3,-2 * \y +  -1.3);
                }
            }
        }
        \newsavebox{\leftpiece}
        \savebox{\leftpiece}{
            \foreach \y in {0, ..., 2}{
                \draw[fill, color=orange!20] (-0.2, -2 * \y + 0.2) rectangle (1.2,  -2 * \y + -1.2);
                \draw[fill, color=orange!20] (2-0.2, -2 * \y +  0.2) rectangle (3.2, -2 * \y +  -1.2);
            }
            \foreach \x in {0, ...,3}{
                \foreach \y in {0, ...,5}{
                    \draw (\x, -\y) node[circle, draw, inner sep=3pt](\x-\y){};
                }
            }
            \foreach \y in {0, ...,5}{
                \draw (0-\y)--(1-\y)--(2-\y)--(3-\y);
            }
            \foreach \x in {0, ...,3}{
            \draw(\x-0) -- (\x-1) -- (\x-2) -- (\x-3) -- (\x-4) -- (\x-5);
            }
            
            \draw (0-2) .. controls (-1,-2.33) and (-1,-3.66) .. (0-4);
            \draw (0-3) .. controls (-1,-3.33) and (-1,-4.66) .. (0-5);
            
            \foreach \y in {0, ..., 2}{
                \foreach \x in {0, ...,3}{
                    \draw[color=blue, thick] (\x - 0.3, -2 * \y + 0.3) rectangle (\x+0.3,-2 * \y +  -1.3);
                }
            }
        }
        \draw(0,0)node{\usebox{\leftpiece}};
        \foreach \n in {1, ..., 5}
        {
            \draw (3.75, -\n + 0.5) node{$\cdots$};
        }
        \draw(4.5,0)node{\usebox{\rightpiece}};
        \foreach \n in {1, ..., 3}
        {
            \draw (\n - 0.5, -5.5) node{$\vdots$};
        }
        \foreach \n in {1, ..., 3}
        {
            \draw (\n - 0.5 + 4.5, -5.5) node{$\vdots$};
        }
        \draw[thick, <->] (0, 0.75) -- node[above]{$2K$} (4.5 + 3, 0.75);
        \draw[thick, <->] (-1.25, 0) -- node[left]{$2L$} (-1.25, -7);
        \end{tikzpicture}}
        \caption{Slightly altered grid layout. Odd (resp. even) rows are connected to the row below on the right (resp. left) via two extra edges. In total, this amounts for $2L-2$ extra edges.
        The Width-2 ladder and Grid block layouts are depicted in blue and red respectively.}\label{fig:extra_edges}
\end{figure}
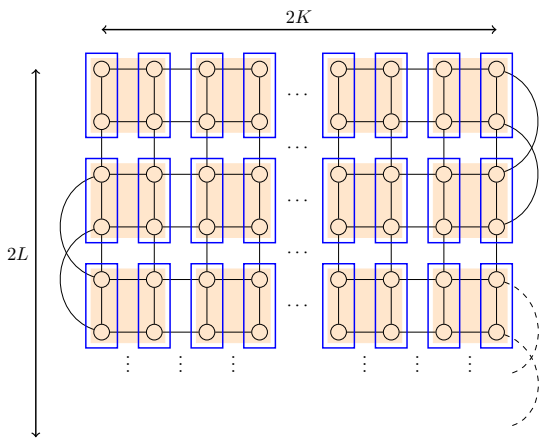

\begin{mypropo}
    In a $2L\times 2K$ grid layout ($n=4LK$) with $2L-2$ additional connections as in Figure \ref{fig:extra_edges}, any $n-$qubits linear reversible operator can be implemented in depth at most $15n/4 + O(1)$.
\end{mypropo}

\begin{mypropo}
    In a $2L\times 2K$ grid layout ($n=4LK$) with diagonals (every other row) with $2L-2$ additional connections as in Figure \ref{fig:extra_edges}, any $n-$qubits linear reversible operator can be implemented in depth at most $7n/2 + O(1)$.
\end{mypropo}

The proofs of those propositions are similar to the one of Proposition \ref{prop:ladder_2}. 
We use the additional connections in order to allow for a Width-2 ladder block layout to perform the second step.

Consider for instance a square grid of even dimensions $n=4L^2$. 
In that setting, by adding $2L-2 = o(n)$ connections, we can obtain an asymptotic improvements of $n/4$ with respect to the complexity obtained for a standard Grid and Grid + diagonals layouts.

\subsection{All-to-all connectivity between larger blocks}

We now explore the more general case where we have a full connectivity between any pair of qubits in neighboring blocks and where the block size is arbitrary. In other words, in that setting, all the qubits in one block are fully connected, and all the qubits in block $b_i$ are connected to all the qubits in block $b_{i+1}$.
This is an hypothesis that can become true in some hardware models where each qubit has an interaction radius: each qubit can interact with any qubit within a certain distance \cite{henriet2020quantum}. In the case of a grid, with suitable interaction radius, our hypothesis can be true and the larger the radius the larger the block size can be. 

We propose two ways to solve the two problems, our method essentially relies on previous works about the synthesis of linear reversible circuits for unconstrained architecture \cite{de2021reducing,maslov2022depth}. 

\subsubsection{Problem 1}

By assumption we are given a full rank $2p \times p$ binary matrix of the form 
\[ \begin{bmatrix} A_1 \\ A_2 \end{bmatrix}. \]
With a CNOT circuit of depth $1$ we can add suitable rows of $A_2$ to $A_1$ to make $A_1$ invertible. Then, we consider the matrix $B = A_2A_1^{-1}$ that we zero with the following available operations: 
\begin{itemize}
    \item elementary row operations on $A_2$ $\to$ available row operations on $B$,
    \item elementary row operations on $A_1$ $\to$ available column operations on $B$,
    \item elementary row operations from $A_1$ to $A_1$ $\to$ flip one entry of $B$. 
\end{itemize}

It is well-known that we can encode $B$ in a bipartite graph such that any matching on $B$ corresponds to a set of parallel row operations between $A_1$ and $A_2$. If $B$ has at most $k$ entries equal to $1$ on each row and column, then the corresponding bipartite graph has degree $k$ and can be decomposed as a sum of $k$ matchings. Therefore, $B$ can be zeroed in depth at most $k$. $B$ is $p \times p$ so we are ensured that $B$ can be zeroed in depth at most $p$.

\begin{myprop} Problem 1 can be solved in depth at most $1+p$.
\end{myprop}

We can improve this result by using a technique used in \cite{maslov2022depth}, they show that we can write 
\[ B = B' \oplus \mathbf{1} \cdot v_1^T \oplus v_2 \cdot \mathbf{1}^T \]
where $v_1, v_2$ are two arbitrary vectors and $B'$ has at most $\floor{p/2}$ entries on each row and column. Therefore step 1 can be realized first by performing $B \oplus B'$ in depth $\floor{p/2}$. Then 
\[  B \oplus B' = \begin{bmatrix} \mathbf{1} &  v_2 \end{bmatrix} \cdot \begin{bmatrix} v_1 &  \mathbf{1} \end{bmatrix}^T  = w_1 \cdot w_2^T \]
where $w_1, w_2$ are $p \times 2$. We can reduce in parallel both $w_1$ and $w_2$ with elementary row and column operations on $B \oplus B'$. Note that only $2$ different nonzero rows can be found in $w_1$: $[1,0]$ and $[1,1]$, similarly in $w_2$ we can only find $[0,1]$ and $[1,1]$. In both cases, we can zero any duplicate in depth at most $\ceil{\log(p)}$: given $m$ occurences of $[1,0]$ for instance, we can zero $\floor{m/2}$ of them in depth $1$, and we repeat the process. We eventually have the top $2 \times 2$ entries of $w_1$ and $w_2$ that are nonzero: this corresponds to a $2 \times 2$ block in $B$ that can be zeroed in depth at most $2$. 

\begin{myprop} Problem 1 can be solved in depth at most $3+\floor{p/2}+\ceil{\log(p)}$.
\end{myprop}

\subsubsection{Problem 2}

In the all-to-all case, Problem 2 is very similar to Problem 1. By assumption we are given a $2p \times 2p$ binary matrix of the form 
\[ \begin{bmatrix} A_1 & A_3 \\ A_2 & 0 \end{bmatrix} \]

where $A_2$ and $A_3$ are invertible. In depth $2$ we can obtain the matrix
\[ \begin{bmatrix} A_2 & 0 \\ A_1 \oplus A_2 & A_3 \end{bmatrix} \]
and we have to zero $A_1 \oplus A_2$ using $A_2$ in a similar way that we did during first step. 

\begin{myprop} Problem 2 can be solved in depth at most $2+p$.
\end{myprop}

\begin{myprop} Problem 2 can be solved in depth at most $4+\floor{p/2}+\ceil{\log(p)}$.
\end{myprop}

\subsubsection{Third step}

For this step we can also rely on the work already done for linear reversible circuits synthesis on unconstrained architectures \cite{de2021reducing,maslov2022depth,jiang2020optimal}. The best method so far is a divide-and-conquer algorithm but for simplicity we use the adaptation of Kutin \textit{et al.}'s algorithm to the all-to-all connectivity: any $p$-qubit linear reversible operator can be synthesized in depth at most $2p+6$.

\subsubsection{Results}

\begin{mythm} In a quantum hardware with blocks of $p$ qubits arranged on a line and with full connectivity between any pair of consecutive blocks, any $n$-qubit linear reversible operator can be synthesized in depth at most $\left(2 + 3/p\right)n+ 2p + 6$.
\end{mythm}

\begin{proof}
We simply add the depths: 
\[ d(n) = n/p \times (1+p+2+p) + 2p+6 = 3n/p + 2n + 2p + 6. \]
\end{proof}

\begin{mythm} In a quantum hardware with blocks of $p$ qubits arranged on a line and with full connectivity between any pair of blocks, any $n$-qubit linear reversible operator can be synthesized in depth at most $\left(1+7/p+2\ceil{\log(p)}/p\right)n + 2p + 6$.
\end{mythm}

\begin{proof}
We simply add the depths: 
\[ d(n) = n/p \times (3+\floor{p/2}+\ceil{\log(p)} + 4+\floor{p/2}+\ceil{\log(p)}) + 2p+6 < n + 7n/p + 2n\ceil{\log(p)}/p + 2p + 6. \]
\end{proof}

\section{Conclusion} \label{sec::conclusion}

We proposed a block generalization of Kutin \emph{et al.}'s algorithm that synthesizes CNOT circuits for an LNN architecture. Our generalization needs the blocks of qubits to be arranged as a line in the hardware. Despite this prerequisite, some realistic quantum hardware can fit into our framework. We showed that the depth complexity of our algorithm essentially depends on the solving of two elementary problems involving $O(p)$-sized boolean matrices where $p$ is the size of the block. We brute-forced the solution for some small blocks ($p \leq 4$) and we gave an algorithm for general $p$ when the blocks are fully connected. As a result we improved the depth complexity for useful quantum hardware such as the grid. 

As a future work, it would be interesting to extend this framework to other classes of circuits: CZ circuits, Clifford circuits, phase polynomials. We could have used normal form for Clifford circuits \cite{aaronson2004improved, maslov2018shorter, duncan2020graph, bravyi2021hadamard} to propose preliminary results but we believe deeper analysis can be done.

\section*{Acknowledgments}

This work has been supported by the French state through the ANR as a part of \emph{Plan France 2030}, projects NISQ2LSQ (ANR-22-PETQ-0006) and EPiQ (ANR-22-PETQ-0007), as well as the ANR project SoftQPro (ANR-17-CE25-0009).
The authors would like to thank Jérôme Pioux for his patience: running the enumerations took a bit longer than anticipated :)

\bibliographystyle{abbrv}
\bibliography{Biblio}

\begin{thebibliography}{10}

\bibitem{aaronson2004improved}
S.~Aaronson and D.~Gottesman.
\newblock Improved simulation of stabilizer circuits.
\newblock {\em Phys. Rev. A}, 70:052328, Nov 2004.

\bibitem{amy2018controlled}
M.~Amy, P.~Azimzadeh, and M.~Mosca.
\newblock On the controlled-{NOT} complexity of
  controlled-{NOT}{\textendash}phase circuits.
\newblock {\em Quantum Science and Technology}, 4(1):015002, sep 2018.

\bibitem{andrews1998theory}
G.~E. Andrews.
\newblock {\em The theory of partitions}.
\newblock Number~2. Cambridge university press, 1998.

\bibitem{bravyi2005universal}
S.~Bravyi and A.~Kitaev.
\newblock Universal quantum computation with ideal clifford gates and noisy
  ancillas.
\newblock {\em Phys. Rev. A}, 71:022316, Feb 2005.

\bibitem{bravyi2021hadamard}
S.~Bravyi and D.~Maslov.
\newblock Hadamard-free circuits expose the structure of the clifford group.
\newblock {\em IEEE Transactions on Information Theory}, 67(7):4546--4563,
  2021.

\bibitem{de2021decoding}
T.~G. de~Brugi{\`e}re, M.~Baboulin, B.~Valiron, S.~Martiel, and C.~Allouche.
\newblock Decoding techniques applied to the compilation of cnot circuits for
  nisq architectures.
\newblock {\em Science of Computer Programming}, page 102726, 2021.

\bibitem{de2021gaussian}
T.~G. De~Brugi{\`e}re, M.~Baboulin, B.~Valiron, S.~Martiel, and C.~Allouche.
\newblock Gaussian elimination versus greedy methods for the synthesis of
  linear reversible circuits.
\newblock {\em ACM Transactions on Quantum Computing}, 2(3):1--26, 2021.

\bibitem{de2021reducing}
T.~G. de~Brugi{\`e}re, M.~Baboulin, B.~Valiron, S.~Martiel, and C.~Allouche.
\newblock Reducing the depth of linear reversible quantum circuits.
\newblock {\em IEEE Transactions on Quantum Engineering}, 2:1--22, 2021.

\bibitem{duncan2020graph}
R.~Duncan, A.~Kissinger, S.~Perdrix, and J.~Van De~Wetering.
\newblock Graph-theoretic simplification of quantum circuits with the
  zx-calculus.
\newblock {\em Quantum}, 4:279, 2020.

\bibitem{gottesman1997stabilizer}
D.~Gottesman.
\newblock {\em Stabilizer codes and quantum error correction}.
\newblock PhD thesis, Caltech, 1997.

\bibitem{henriet2020quantum}
L.~Henriet, L.~Beguin, A.~Signoles, T.~Lahaye, A.~Browaeys, G.-O. Reymond, and
  C.~Jurczak.
\newblock Quantum computing with neutral atoms.
\newblock {\em Quantum}, 4:327, 2020.

\bibitem{jiang2020optimal}
J.~Jiang, X.~Sun, S.-H. Teng, B.~Wu, K.~Wu, and J.~Zhang.
\newblock Optimal space-depth trade-off of cnot circuits in quantum logic
  synthesis.
\newblock In {\em Proceedings of the Fourteenth Annual ACM-SIAM Symposium on
  Discrete Algorithms}, pages 213--229. SIAM, 2020.

\bibitem{DBLP:journals/qic/KissingerG20}
A.~Kissinger and A.~M. de~Griend.
\newblock {CNOT} circuit extraction for topologically-constrained quantum
  memories.
\newblock {\em Quantum Inf. Comput.}, 20(7{\&}8):581--596, 2020.

\bibitem{knill2005quantum}
E.~Knill.
\newblock Quantum computing with realistically noisy devices.
\newblock {\em Nature}, 434(7029):39--44, 2005.

\bibitem{knill2008randomized}
E.~Knill, D.~Leibfried, R.~Reichle, J.~Britton, R.~B. Blakestad, J.~D. Jost,
  C.~Langer, R.~Ozeri, S.~Seidelin, and D.~J. Wineland.
\newblock Randomized benchmarking of quantum gates.
\newblock {\em Physical Review A}, 77(1):012307, 2008.

\bibitem{DBLP:journals/cjtcs/KutinMS07}
S.~A. Kutin, D.~P. Moulton, and L.~Smithline.
\newblock Computation at a distance.
\newblock {\em Chicago J. Theor. Comput. Sci.}, 2007, 2007.

\bibitem{magesan2011scalable}
E.~Magesan, J.~M. Gambetta, and J.~Emerson.
\newblock Scalable and robust randomized benchmarking of quantum processes.
\newblock {\em Physical review letters}, 106(18):180504, 2011.

\bibitem{maslov2007linear}
D.~Maslov.
\newblock Linear depth stabilizer and quantum fourier transformation circuits
  with no auxiliary qubits in finite-neighbor quantum architectures.
\newblock {\em Phys. Rev. A}, 76:052310, Nov 2007.

\bibitem{maslov2018shorter}
D.~Maslov and M.~Roetteler.
\newblock Shorter stabilizer circuits via bruhat decomposition and quantum
  circuit transformations.
\newblock {\em {IEEE} Trans. Inf. Theory}, 64(7):4729--4738, 2018.

\bibitem{maslov2022depth}
D.~Maslov and B.~Zindorf.
\newblock Depth optimization of cz, cnot, and clifford circuits.
\newblock {\em arXiv preprint arXiv:2201.05215}, 2022.

\bibitem{nash2020quantum}
B.~Nash, V.~Gheorghiu, and M.~Mosca.
\newblock Quantum circuit optimizations for {NISQ} architectures.
\newblock {\em Quantum Science and Technology}, 5(2):025010, 2020.

\bibitem{patel2008optimal}
K.~N. Patel, I.~L. Markov, and J.~P. Hayes.
\newblock Optimal synthesis of linear reversible circuits.
\newblock {\em Quantum Information \& Computation}, 8(3):282--294, 2008.

\bibitem{tang2020efficient}
Y.~Tang.
\newblock Efficient cnot synthesis for nisq devices.
\newblock {\em arXiv preprint arXiv:2011.06760}, 2020.

\end{thebibliography}

\end{document}